\documentclass[11pt]{amsart}

\usepackage{adjustbox,array}
\usepackage{fullpage}
\usepackage{amsmath, amssymb, amsthm}
\usepackage{booktabs}
\usepackage{cleveref}
\usepackage{comment}
\usepackage[foot]{amsaddr}
\usepackage{algorithm,algpseudocode}
\usepackage{graphicx,subcaption}
\usepackage{multirow}
\usepackage{tikz}
\usetikzlibrary{cd}
\usepackage{url}

\newcommand{\Prob}[2]{\mathbb{P}_{#1}\left( #2 \big.\right)}

\graphicspath{{./figs/}} 

\newcommand{\creflastconjunction}{, and~}

\newtheorem{theorem}{Theorem}

\newtheorem{corollary}[theorem]{Corollary}

\newtheorem{lemma}[theorem]{Lemma}

\newtheorem{proposition}[theorem]{Proposition}
\newtheorem{remark}[theorem]{Remark}

\numberwithin{equation}{section}
\numberwithin{theorem}{section}

\title[Targeting influence in a harmonic opinion model]{Targeting influence in a harmonic opinion model}\thanks{Z.M.B. and P.J.M. were supported by ARO MURI award W911NF-18-1-0244. 
Z.M.B. was also supported by NSF DMS-2137511.
J.L.M. acknowledges support from the NSF through grant DMS-2307384 and FRG grant DMS-2152289. 
P.J.M. was also supported by the NSF through grant BCS-2140024. 
B.O. acknowledges support from NSF DMS-1752202 and DMS-2136198.}

\author[Z.~M.~Boyd]{Zachary~M.~Boyd} \address{Department of Mathematics, Brigham Young University, Provo, UT 84602, USA} \email{zachboyd@byu.edu}

\author[N.~Fraiman]{Nicolas Fraiman} \address{Department of Statistics and Operations Research, University of North Carolina at Chapel Hill, Chapel Hill, NC 27599, USA} \email{fraiman@email.unc.edu}

\author[J.~L.~Marzuola]{Jeremy~L.~Marzuola} \address{Department of Mathematics, University of North Carolina at Chapel Hill, Chapel Hill, NC 27599, USA} \email{marzuola@math.unc.edu}

\author[P.~.J.~Mucha]{Peter~J.~Mucha} \address{Department of Mathematics, Dartmouth College, Hanover, NH 03755, USA} \email{peter.j.mucha@dartmouth.edu}

\author[B.~Osting]{Braxton Osting} 
\address{Department of Mathematics, University of Utah, Salt Lake City, UT 84112, USA} 
\email{osting@math.utah.edu}

\begin{document}

\maketitle

\begin{abstract}
Influence propagation in social networks is a central problem in modern social network analysis, with important societal  applications in politics and advertising. A large body of work has focused on cascading models, viral marketing\creflastconjunction finite-horizon diffusion. There is, however, a need for more developed, mathematically principled \emph{adversarial models}, in which multiple, opposed actors strategically select nodes whose influence will maximally sway the crowd to their point of view.

In the present work, we develop and analyze such a model based on harmonic functions and  linear diffusion. 
We prove that our general problem is NP-hard and that the objective function is monotone and submodular; consequently, we can greedily approximate the solution within a constant factor.
Introducing and analyzing a convex relaxation, we show that the problem can be approximately solved using smooth optimization methods. 
We illustrate the effectiveness of our approach on a variety of example networks. 
\end{abstract}



\section{Introduction} \label{sec:Intro}
On social media platforms (and in real life), ideas, information\creflastconjunction opinions propagate through connections in the social network; network members may influence their connections, 
e.g., swaying the political viewpoint of friends, 
disseminating information about current events, or 
increasing product awareness in followers. 
To be concrete, let's consider a particular product for which there are $k=2$ competing brands: a blue brand and a red one. 
Each member of the network will have some degree of preference between the brands and also have varying influence on their connections in the network. 
Indeed, \emph{influencers}  emerge on social media platforms and have the ability to widely market their preferred color product to their many connections. 
Additionally, a \emph{zealot} is an individual that is uncompromising or inflexible in their (typically extreme) idea/opinion.
An \emph{authority} for the blue or red brand wants to most broadly advertise it within the network. They may not be a node in the network themselves, but instead try to influence the network by persuading key nodes to join their side. Companies and political parties are examples of authorities.
We suppose that authorities have a budget to strategically target members of the network and convert them to become zealots. For example, they may sponsor them or otherwise incentivize them to disseminate their viewpoint. 
Intuitively, an authority might target members of the network with many connections, but it is not so simple. For example, choosing a powerful influencer may not be very effective if most of their connections already agree with the authority's point of view.
We refer to the problem of authorities targeting members within a social network so as to maximize the influence of a particular opinion as the \emph{targeting influence problem}. 
The targeting influence problem appears in many guises beyond the targeted product marketing example given here; indeed, swaying public opinion is a natural goal in both politics and advertising. As an additional example, it is paramount for various authorities to combat ``fake news" by disseminating factual information to influential members of a network.

\subsection{Contributions and results} In this paper, \emph{we introduce and analyze a model for targeting influence in a social network in settings where there are a multitude of extreme opinions.} 
The opinion of each member of the network lies in an opinion space, taken to be the convex hull of some extreme opinions; see \cref{s:OpinionModel}. The opinion of each individual is influenced by their (directed) connections in the social network. We use a simple, harmonic model based on DeGroot learning \cite{DeGroot_1974} and label propagation \cite{Zhou_2005} to describe these opinion dynamics; see \cref{s:OpinionDynamics}. 
In \cref{s:authority}, we introduce a measure for the influence for a particular opinion and an optimization problem that describes how an authority might target influence within the social network; see \cref{e:OptProb}. The measure of influence of each opinion is simply the sum of the opinions of the individuals in the network. 
In \cref{s:binary}, we show that 
each authority can group together the opposing opinions/ideas when considering how to maximize their own influence, thus resulting in a reduced \emph{grouped opponent targeting influence problem}; see \eqref{e:OptProb2}.
In \cref{s:probInterp}, we review the probabilistic interpretation of harmonic functions on graphs in terms of the hitting time of an associated random walk. 
In \cref{s:energyInterp}, we describe how the problem, in the case of an undirected 
network, can be interpreted in terms of the graph Dirichlet energy.

In \cref{sec:HandA}, we prove  that there exists a $\varepsilon > 0$ such that it is NP-hard to obtain a $(1-\varepsilon)$-approximation to the targeting influence  problem; see \cref{t:NP-hard}. Our proof uses a gap-introducing reduction to the vertex cover problem on cubic (3-regular) graphs, which is APX-hard. To recall, APX problems are optimization problems that have a polynomial-time constant-factor approximation algorithm. Some problems in APX have a polynomial-time approximation scheme (PTAS) and can be approximated to any factor. However, the PCP Theorem (probabilistically checkable proofs) \cite{arora1998pcp} established that there are problems in APX which cannot be approximated efficiently by a PTAS. 
These so-called APX-hard problems have approximation-preserving reductions from every problem in APX. Therefore, APX-hard problems do not have a PTAS and cannot be approximated to an arbitrary factor. 
We establish that the objective function in the influence targeting problem is monotone and submodular and hence admits a $1 - 1/e$ approximation algorithm; see \cref{t:submodular}. Our proof relies on the probabilistic interpretation of harmonic functions on graphs. 

In \cref{sec:Relax}, we develop an approximation strategy for solutions of the  NP-hard targeting influence  problem. 
 We show that the objective function for the relaxed optimization  problem \eqref{e:OptProb2Relax} is concave and compute the gradient and Hessian; see \cref{t:Convexity}. 
In \cref{sec:GrAut}, we 
consider the relaxed optimization problem and prove that, in the case when the graph has a symmetry, the optimal solution also has the same symmetry; see \cref{p:GraphAuto}.

In \cref{sec:Num}, we describe a numerical implementation of our method and present the outcomes of numerous experiments to illustrate our results and test how well the relaxed optimization problem predicts the solution to the original (APX-hard) problem.  
For simplicity, we consider a game between two opinion authorities to capture the largest proportion of influence, where each authority alternatively selects one vertex to convert to their opinion. 
In general, we see that the relaxed problem highlights similar features of the graph as the unrelaxed problem giving a reasonable and computationally efficient strategy for targeting influence.   

We conclude in \cref{sec:Disc} with a discussion.

\subsection{Previous work}
There are a wide variety of interrelated opinion dynamics, belief propagation\creflastconjunction consensus formation models that have been introduced and analyzed. Generally in these models, the opinions of members in a social network evolve in time via network interactions. The member opinions can be modeled as discrete or continuous and can account for more than two opinions. 
The interactions between members can be modeled using Bayesian or non-Bayesian approaches; 
here we focus on non-Bayesian approaches, which can further be characterized as dynamical/stochastic models 
in discrete/continuous time. 
A review of this vast literature is beyond the scope of this paper (see, e.g., \cite{Acemoglu2011,peralta2022opinion}), 
but non-Bayesian opinion dynamics models include or are motivated by   
cascading models \cite{Kempe_2003,Kempe_2005,Kempe_2015}, 
percolation models, 
the Vicsek swarming/flocking model \cite{Vicsek_1995}, 
the Kuramoto synchronization model \cite{Kuramoto_1975}, 
the DeGroot model \cite{DeGroot_1974} and its many variants (as described in, e.g., \cite{Brooks_2024},
the label propagation model \cite{Zhou_2005},
voter models \cite{Redner_2019}\creflastconjunction
bounded confidence models \cite{Bernardo_2024}). 
In the present work, we consider the steady state of a specific form of the DeGroot dynamical model, which is closely related to the label propagation model.

Within opinion dynamics models a wide variety of questions are posed and studied.
What is the steady state of the dynamics: does the population reach consensus or are opinions polarized/fragmented? 
How do opinion dynamics depend on the topology of the social network 
or on the way in which interactions between members are modeled? 
How do we identify authoritative members of the network \cite{Kleinberg_2011,Pei_2019}? 
If influencers/zealots/opinion authorities/forceful agents are introduced, how does this effect opinion dynamics \cite{Acemoglu_2009,Galam_2007,Mobilia_2007,verma2014impact} and how does this change with network topology \cite{klamser2017zealotry}? 
We view the targeting influence problem as an optimal control/inverse problem for this last question.

There are several previous works that look at targeting influence or ``influence maximization,'' as considered in this paper. To our knowledge, this type of question was first mathematically posed in \cite{Domingos_2001} who modeled influence by a Markov random field. Subsequent work also analyzed this type of question when opinion dynamics is an independent cascade model \cite{bharathi2007competitive,Mossel_2010}.

The  paper most similar to ours is a paper by Fardad, Zhang, Lin\creflastconjunction Jovanov\'ic \cite{fardad2012optimal}, where the authors formulate a binary opinion/idea targeting influence problem equivalent to our grouped opponent targeting influence problem \eqref{e:OptProb2}. In particular, the opinion dynamics model is taken to be the steady-state of the DeGroot model. This problem is reformulated and approximately solved via the alternating direction method of multipliers (ADMM). 
Interestingly, the authors also consider perturbations to the network connections to enhance the influence of zealots.

In \cite{zhao2015competition,zhao2014competitive}, the authors also consider a targeting influence problem, where the opinion dynamics model is also taken to the steady-state of the DeGroot model. However, whereas we consider the  measure of influence of each opinion as the sum of the opinions of the individuals in the network, \cite{zhao2015competition,zhao2014competitive} introduces a measure which counts the voters who prefer each opinion. This ``competition'' between binary opinions/ideas might be more realistic in an election, for example. The authors introduce an  Influence Matrix (IM) criterion to predict the bias of each non-zealot individual and hence the result of the competition. They compare the IM criterion with seven centrality-based criteria and extensive numerical experiments suggest that 
the IM criterion is effective at predicting the competition result.

\subsection*{Notation} 
Write $[k] = \{ 1,2,\ldots,k\}$ for $k \in \mathbb N\setminus \{0\}$. 
Denote by $e_\ell \in \mathbb R^k$ the characteristic vector for coordinate $\ell \in [k]$. 
Denote the unit simplex by $\Delta_k = \{u \in \mathbb R^k\colon u\geq 0, \ \sum_{\ell \in [k]} u_\ell = 1\} = \textrm{conv}\left( \{e_\ell\}_{\ell \in [k]} \right)$. 
For vectors or matrices $a$ and $b$, $a\odot b$ represents the Hadamard product. The outer product of two vectors $a$ and $b$ is written $a\otimes b$.

\section{Opinion dynamics and targeting influence} \label{sec:Models}

\subsection{A model for opinions and influencers} \label{s:OpinionModel}
We consider a particular topic and assume that there are $k\geq 2$ (fixed) extreme opinions on the topic. 
Each individual (represented by a vertex in the network) has an opinion on the topic, and we model each individual's opinion as a convex combination of the extreme opinions. 
To this end, we represent the extreme opinions by the coordinate vectors $e_\ell \in \mathbb R^k$. 
The \emph{opinion of individual $i$} is represented as the vector $u(i) = (u_1, \ldots, u_k)(i) \in \Delta_k$.
Thus, the \emph{opinion space} is the unit simplex, $\Delta_k \subset \mathbb R^{k}_+$, the convex hull of the extreme opinions.

\subsection{A harmonic model for opinion dynamics} \label{s:OpinionDynamics}
Here we describe a simple model for opinion dynamics; this model is related to the DeGroot model \cite{DeGroot_1974} and the label propagation model \cite{Zhou_2005}, as well as other models in opinion dynamics. We call attention to the introduction of \cite{Brooks_2024}, which delineates the literature on numerous variants of the modeling we leverage here, including the introduction of continuous-time modeling \cite{Abelson_1967} and the inclusion of zealots \cite{Friedkin_1990,Taylor_1968}, as well as many other models with rich behaviors. For our present purposes, we introduce the model we consider through a continuous-time dynamics; but we are ultimately focused only on the steady states resulting from that model in formulating our influence maximization problem (though, in so doing, we recognize that models yielding qualitatively different steady states might also be interesting to study in this manner).
We will assume that the individuals are connected via a social network. Within the network, there are zealots with (fixed) extreme opinions and other members with opinions that are determined by their social connections. Roughly speaking, the opinion of the non-zealot individuals will be the average of the opinions of the individuals networked to the individual. 

Let $G = (V,E)$ be a strongly connected, directed graph that represents the social network.  
Set $n = |V|$ and enumerate the nodes so that we may identify  $V$ with $[n]$. 
Let $Z_\ell \subset V$, $\ell \in [k]$ be disjoint vertex sets representing zealots, where $Z_\ell$ is the set of zealots for opinion $\ell \in [k]$. Consequently, we set 
$$
u(i) = e_\ell,  \qquad \qquad i \in Z_\ell, \ \ell \in [k].
$$
The vertex subset $Z := \amalg_{\ell \in [k]} Z_\ell$ is the \emph{collection of all zealots}.
Note that for every vertex $i \in Z$, $u(i)$ is the characteristic vector for the disjoint partition $Z = \amalg_{\ell \in [k]} Z_\ell$. 

Since $G$ is a directed graph, individual $i$ may be influenced by $j$ even if $j$ is not influenced by $i$. We use the convention that a directed edge $(i,j)$ represents that $i$ is influenced by $j$. That means edges are directed in the opposite direction in which opinions spread in the network. 
Let $A\in \mathbb R^{n \times n}$ be the directed adjacency matrix, 
$$
A_{i,j} = \begin{cases}
1 & \textrm{$i$ is influenced by $j$} \\ 
0 & \textrm{otherwise}
\end{cases}.
$$
Let 
$d_i = \sum_j A_{ij}$ be the out-degree of vertex $i \in V$ and define the diagonal matrix, $D \in \mathbb R^{n\times n}$ of out-degrees by $D_{ii} = d_i$, with $D_{ij}=0$ whenever $i\neq j$.
Denote the (out-degree) graph Laplacian for $G$ by $L = D-A$.  In the above construction, influence along edges is considered binary (edge weights are 1 or 0), but our analysis is unaffected if we instead consider weighted adjacency matrices where influence can be weighted differently along edges, i.e. $A=[a_{ij}]$ for $a_{ij} \geq 0$, and the corresponding weighted graph Laplacian $L=D-A$ where now $d_i = \sum_j a_{ij}.$

We assume a specific case of the DeGroot model, wherein the opinions evolve according to linear diffusion, satisfying 
\begin{subequations} \label{e:dynamics}
\begin{align}
& \frac{d}{dt}  u(i)  = -  (L u)(i) = - \sum_j L_{ij}u(j),  && i \in V 
\setminus Z, \label{e:dynamicsA}\\
& u(i) = e_\ell, && i \in Z_\ell, \ \ell \in [k], \label{e:dynamicsB}
\end{align}
\end{subequations}
with appropriate initial conditions. As $t \to \infty$, the opinions approach a steady state. 
The opinions of the non-zealot vertices are then defined to equal this steady state, satisfying
\begin{subequations} \label{e:harmo}
\begin{align}
& L u(i)  = 0,  && i \in V \setminus Z, \label{e:harmoA} \\ 
& u(i) = e_\ell, && i \in Z_\ell, \ \ell \in [k]. \label{e:harmoB}
\end{align}
\end{subequations}
Note that equation \eqref{e:harmoA} is written using vector notation assuming the order of operations that applies the Laplacian before evaluation for the $i$th node and is performed separately for each opinion index; that is, $Lu(i)$ should be read as $(L u_\ell)(i) = 0$ for each $\ell \in [k]$. 
We interpret $u$ to be a harmonic (vector-valued) vertex function satisfying Dirichlet boundary conditions, i.e., a  \emph{graph Laplace problem}. 
In particular, since $L = D-A$, we have that non-zealot $i$'s opinion is the average of their outgoing neighbors $N(i)$ opinions, 
$$
u(i) = \frac{1}{d_i} \sum_{j \in N(i)} u(j), \qquad \qquad i \in V \setminus S,
$$
where $d_i$ is the out-degree of node $i$. 
Because of this local averaging, $u(i)$ is necessarily in the opinion space $\Delta_k$ if its neighbors' opinions are also in that opinion space.
In \eqref{e:harmo}, the zealots (vertices in the boundary $Z$) have opinion in the opinion space $\Delta_k$. 
The following lemma shows that $u(i)$ for every $i \in V$ is in the opinion space $\Delta_k$. 

\begin{lemma} \label{l:HarmonicFact}
Let $Z \subset V$  and $f\colon Z \to \Delta_k$. 
Suppose $u \colon V \to \mathbb R^k$ solves 
\begin{align*}
& L u(i)  = 0,  && i \in V \setminus Z, \\ 
&u(i) = f(i), && i \in Z. 
\end{align*}
Then $u(i) \in \Delta_k$ for every $i \in V$. 
\end{lemma}

\begin{proof}
Looking at each component $\ell \in [k]$ individually,
by the graph maximum principle, we have that $u_\ell(i)\geq 0$ for all $i \in V$ and $\ell \in [k]$. Defining $v = \sum_{\ell \in [k]} u_\ell$, we have that $v \colon V \to \mathbb R$ solves the Laplace problem 
\begin{align*}
& L v(i)  = 0,  && i \in V \setminus Z, \\ 
&v(i) = 1, && i \in Z.
\end{align*}
The unique solution is given by $v=1$, which gives that $u(i) \in \Delta_k$ for every $i \in V$.
\end{proof}

As an aside, we note that the simplex opinion space constraint holds not only for the steady state solution but also for the underlying dynamics in \eqref{e:dynamics} when all individual opinions are initiated on the unit simplex. 

Limitations of the harmonic model for opinion dynamics in the context of the targeting influence problem  will be discussed in \cref{sec:Disc}.

\subsection{Measuring and targeting influence} \label{s:authority}  
Suppose that we have opinions given by $u\colon V \to \Delta_k$. We measure the influence of opinion $m \in [k]$ by 
\begin{equation}
\label{e:Jm}    
\mathcal{I}_m(u) := \frac{1}{|V|} \sum_{i \in V} u_m (i) = \frac{1}{|V|} \|u_m \|_{\ell^1(V)}. 
\end{equation}
Note that the total influence of all opinions is given by 
$$
\sum_{m \in [k]} \mathcal{I}_m (u) 
= \frac{1}{|V|} \sum_{i \in V} \sum_{m \in [k]}u_m (i)
= 1.
$$
We thus interpret $\mathcal{I}_m(u)$ to be \emph{the proportion of influence that opinion $m \in [k]$ has on the network}. 

We can now formulate the targeting influence problem. 
Let the current zealot set be given by $Z = \amalg_{\ell \in [k]} Z_\ell$. 
Consider the authority for opinion  $m \in [k]$.  We suppose that this authority wants to augment their zealot set $Z_m \to Z_m \cup T$ such that $T \subset V\setminus Z$ has fixed size $|T| = t$ (corresponding to a fixed budget) as to maximize $\mathcal{I}_m$, the proportion of influence that opinion $m$ has on the network. 
We formulate the \emph{targeting influence problem} as  
\begin{subequations}
\label{e:OptProb}
\begin{align}
\label{e:OptProbA}
\max_{T \subset V}  \ & \mathcal{I}_m(u) \\
\label{e:OptProbB}
\textrm{s.t.} \ &  L u(i)  = 0,  && i \in V \setminus (Z \cup T), \\ 
\label{e:OptProbC}
 & u(i) = e_\ell, && i \in Z_\ell, \ \ell \in [k], \\
 \label{e:OptProbD}
 & u(i) = e_m, && i \in T, \\ 
 \label{e:OptProbE}
 & |T| = t, \ T \cap Z = \varnothing. 
\end{align}
\end{subequations}
Intuitively, if the authority wants to have a large effect on the population, it should choose $T \subset V \setminus Z$ to consists of influencers within the population.

Limitations in our formulation of the targeting influence problem 
\eqref{e:OptProb} will be discussed in \cref{sec:Disc}.

\subsection{Reduction to a grouped opponent targeting influence problem} \label{s:binary}
We now explain how, for fixed authority $m \in [k]$,  the targeting influence problem \eqref{e:OptProb}  
reduces to an us-vs-them targeting influence problem. 

We observe that  $\mathcal{I}_m$, defined in \eqref{e:Jm}, only depends directly on $u_m$ and not $u_\ell$, $\ell \neq m$. Additionally, writing $v \equiv u_m$, the constraint set also separates and the equation for $v\colon V \to [0,1]$ can be written 
\begin{align*}
& L v(i)  = 0, && i \in V \setminus (Z \cup T),  \\ 
& v(i) = 1, && i \in Z_m \cup T, \\
& v(i) = 0, && i \in Z_\ell, \ \ell \neq m. 
\end{align*} 
These three equations replace the constraints \eqref{e:OptProbB}-\eqref{e:OptProbD}. Abusing notation we write $\mathcal{I}_m(v) = \frac{1}{|V|} \|v\|_{\ell^1(V)}$. For fixed $m \in [k]$,  \eqref{e:OptProb} can be equivalently expressed as 
\begin{subequations}
\label{e:OptProb2}
\begin{align}
\label{e:OptProb2A}
\max_{T\subset V}  \ & \mathcal{I}_m(v) \\
\label{e:OptProb2B}
\textrm{s.t.} \ & L v(i)  = 0, && i \in V \setminus (Z \cup T), \\ 
\label{e:OptProb2C}
& v(i) = 1, && i \in Z_m \cup T,  \\
\label{e:OptProb2D}
& v(i) = 0, && i \in Z_\ell, \ \ell \neq m , \\
\label{e:OptProb2F}
 & |T| = t, \ T \cap Z = \varnothing.
\end{align}
\end{subequations}
Thus, authority $m$ can group together the opposing opinions/ideas when consider how to maximize their own influence. We refer to \eqref{e:OptProb2} as the \emph{grouped opponent  targeting influence problem}.
The advantage of \eqref{e:OptProb2} over \eqref{e:OptProb} is that it involves a scalar-valued field on $V$ rather than a vector-valued field on $V$.

\begin{remark}
Equation \eqref{e:OptProb2} is equivalent to \cite[Eq.~(1)]{fardad2012optimal}, but presented with different language. 
\end{remark}

\subsection{Probabilistic interpretation} \label{s:probInterp}

The network opinion can be related to the hitting probabilities of the zealot sets for a random walk. Let $X_t$ be a random walk on the network, that is $\Prob{}{X_{t+1} = j \mid X_t = i} = A_{ij}/d_i$. We write $\Prob{i}{\,\cdot\,} = \Prob{}{\,\cdot\,|\, X_0 = i}$ as shorthand for the random walk starting from node $i$. For a subset of vertices $U$ define its hitting time as $\tau(U) = \inf\{t\geq 0\colon  X_t \in U\}$ for a given random walk realization. That is, importantly, the $\tau(U)$ is a random variable dependent on the $X_t$ random walk and selected subset $U$.

\begin{lemma}\label{l:probInterp}
Given zealot sets $Z_\ell$ for $\ell \in [k]$, if $u$ is the solution to \eqref{e:harmo} we have
\[
\textstyle u_m(i) = \Prob{i}{\tau(Z) = \tau(Z_m)}.
\]
\end{lemma}

\begin{proof}
All we need to check is that the right hand side is harmonic and satisfies the boundary conditions. For $i \in V\setminus Z$, conditioning on the first step of the random walk (using the Markov property and $\Prob{i}{X_1 = j} = 1/d_i$ for $j\in N(i)$) we have
\begin{align*}
\Prob{i}{\tau(Z) = \tau(Z_m)} 
&= \sum_{j\in N(i)} \Prob{i}{\tau(Z) = \tau(Z_m) \mid X_1 = j} \Prob{i}{X_1 = j} \\ 
&= \frac{1}{d_i} \sum_{j\in N(i)} \Prob{j}{\tau(Z) = \tau(Z_m)}.
\end{align*}
Moreover, for the boundary nodes we have
\[
\textstyle \Prob{i}{\tau(Z) = \tau(Z_m)} = \begin{cases}
1 &\text{if } i\in Z_m, \\
0 &\text{if } i\in Z\setminus Z_m.
\end{cases}
\]
Therefore, $\big(\Prob{i}{\tau(Z) = \tau(Z_\ell)}\big)_{\ell\in [k]}$ satisfies equation \eqref{e:harmo}.
\end{proof}

\begin{remark}
Note that $\tau(Z) = \min_{\ell\in [k]} \tau(Z_\ell)$. For fixed $m$ define $Z' = \cup_{\ell\neq m} Z_\ell$. We have that the following two events coincide $\big\{ \tau(Z) = \tau(Z_m) \big\} = \big\{ \tau(Z_m) < \tau(Z') \big\}$.
\end{remark}

One way to interpret this result is to say that $i$ does a random walk until it meets one of the zealot sets and it then picks the corresponding opinion. Then, $u$ is the distribution of the opinion of node $i$.

\subsection{Energy interpretation in the undirected case} \label{s:energyInterp}
We consider the case in which the graph is undirected, i.e., $A_{ij} = A_{ji}$. 
We define the Dirichlet energy of the (vector-valued) graph vertex function $u \colon V \to \Delta_k$ by 
$E(u) = \frac{1}{2} \sum_{\ell \in [k]} \langle u_\ell, L u_\ell \rangle$.
The opinions of the non-zealot vertices are then defined as the solution to the energy minimization problem, 
\begin{subequations} \label{e:Energy}
\begin{align}
\textrm{argmin}_{u} \ & E(u) \\ 
\textrm{s.t.} 
\ &u(i) = e_\ell, && i \in Z_\ell, \ \ell \in [k]. \label{e:EnergyB}
\end{align}
\end{subequations}
The Dirichlet conditions \eqref{e:EnergyB} state that the opinions of the zealots are fixed at their extreme values.  The opinions of non-zealots take values as to minimize the Dirichlet energy of the opinion field $u$. Mathematically, \eqref{e:Energy} defines a harmonic field with image in the unit simplex.
Again, the energy function and constraints in \eqref{e:Energy} are separable, and we can solve for each component of the opinion field $u \colon V \to \Delta_k$ independently.

\section{Hardness and approximation  of the targeting influence problem} \label{sec:HandA}

\subsection{Hardness of the targeting influence problem}

Here we show that the targeting influence problem is NP-hard to approximate by a reduction to the vertex cover problem on cubic (3-regular) graphs.

\begin{remark}
    Establishing the hardness of the problem is most interesting for $k$ growing with $n$. Our problem (and the Vertex Cover problem) is fixed-parameter tractable, i.e., for bounded $k$ one could solve by doing an exhaustive search of all subsets of size $k$ which is polynomial in $n$ (though only really feasible for very small $k$).
\end{remark}

\begin{theorem} \label{t:NP-hard}
There exists an $\varepsilon > 0$ such that it is NP-hard to obtain a $(1-\varepsilon)$-approximation to the influence targeting problem.
\end{theorem}
\begin{proof}[Proof of \cref{t:NP-hard}]
We will use a gap-introducing reduction to the vertex cover problem in cubic graphs (where all vertices have degree exactly $3$). Since this problem is APX-hard \cite{Alimonti_2000} by the PCP Theorem \cite{arora1998pcp} there is a $\delta > 0$ and $k=k(n)$ so that it is NP-hard to decide the promise problem of whether a cubic graph with $n$ vertices (and $m=3n/2$ edges) has a vertex cover of size $\leq k$ (a \texttt{Yes}-instance) or if any subset of $k$ vertices leaves at least $\delta m$ edges uncovered (a \texttt{No}-instance).

Suppose you are given a cubic graph $G=(V,E)$ on $n$ vertices that is a \texttt{Yes/No}-instance. Extend the graph $G$ to $G'$ by adding an extra node $z$ and edges from each vertex $v$ of $G$ to $z$. Consider the influence maximization problem $\max_A J(A)$ with $|A|=k$ and $J(A) = \frac{1}{n+1}\sum_{i\in V} v(i)$ where $Lv(i) = 0$ for all $i\in V\setminus A$, $v(z)=0$ and $v(i)=1$ for all $i\in A$.

Note that since $G$ is $3$-regular any vertex cover has size at least $m/3 = n/2$ and at most $n$. Therefore we can assume $k=(1+\gamma)n/2$ for $\gamma\in [0,1]$. 

If $G$ is a \texttt{Yes}-instance then by picking $A$ to be the vertex cover on $k$ vertices we can lower bound the optimal value by $J(A)$. For each $i \in V\setminus A$ we have $v(i) = \Prob{i}{\tau(A) \leq \tau(z)} = 3/4$. This is because $G$ is 3-regular and each one of those edges has its other endpoint in $A$ and we added an edge to the vertex $z$, so the random walk hits $A$ before $z$ in one step with probability $3/4$. Therefore $\sum_{i\in V} v(i) = k + (3/4)(n-k)$. Using that $k=(1+\gamma)n/2$ we have that $(n+1) J(A) = (1+\gamma)n/2 + (3/4)(1-\gamma)n/2$ which simplifies to give 
\[
\text{OPT} \geq J(A) = \left(\frac{7+\gamma}{8}\right)\frac{n}{n+1}.
\]

If $G$ is a \texttt{No}-instance then any subset $A$ of $k$ vertices leaves at least $\delta m = \delta 3n/2$ edges uncovered. Since $G$ is $3$-regular, there must be at least $\delta n$ vertices with some uncovered edge. This is because there are $\delta 3n$ endpoints of the uncovered edges, if there were fewer than $\delta n$ vertices by the pigeonhole principle one vertex would have to cover more than $3$ of these endpoints which is impossible by the $3$-regularity of $G$.

 If $i$ is one of these $\delta n$ vertices then it has a neighbor $j$ that is not in $A$, the random walk can move from $i$ to $j$ and then $z$, thus $v(i) = \Prob{i}{\tau(A) \leq \tau(z)} \leq 1-1/4-(1/4)^2 = 11/16$. Then, for any $A$ we have that $(n+1) J(A) \leq k + (3/4)(n-k-\delta n) + (11/16)\delta n$. Therefore, using $k=(1+\gamma)n/2$ we can bound
\[
\text{OPT} = \max_A J(A) \leq \left(\frac{7+\gamma}{8}-\frac{\delta}{16}\right) \frac{n}{n+1}.
\]

If we were able to approximate the optimal value with arbitrary accuracy we would be able to determine if the graph $G$ is a \texttt{Yes} or a \texttt{No} instance. 
\end{proof}

\subsection{Approximation of the targeting influence problem}
For fixed $m \in [k]$ and a zealot set $Z = \amalg_{\ell \in [k]} Z_\ell$, we consider function $F_m \colon 2^{V\setminus Z} \to \mathbb R$, which takes a vertex subset $T\subset V \setminus Z$ and returns the value 
$$
F_m(T) := \mathcal{I}_m(v_T) = \frac{1}{|V|} \| v_T \|_{\ell^1(V)},
$$ where $v_T$ satisfies \eqref{e:OptProb2B}-\eqref{e:OptProb2D}. 
We will prove that $F_m$ is a monotone and submodular set function. 
Recall that $f$ is a \emph{monotone function} if for any $T_1,T_2 \subseteq V\setminus Z$ satisfying $T_1 \supseteq  T_2$, we have $f(T_1) \geq f(T_2)$. We say that $f$ is a 
\emph{submodular set function} if for all $T \subseteq V\setminus Z$ and $x,y \in V\setminus Z$, 
$$
f(T\cup \{x\}) - f(T) \geq f(T\cup \{x,y\} ) - f(T\cup \{y\}).
$$
Intuitively, we think of submodularity of a function as giving diminishing returns as elements (vertices) are added to the set $T$. 

For $T \subseteq V\setminus Z$, write $v_i(T)$ as the solution to the following equation evaluated at vertex $i \in V$
\begin{align*}
& L v(i)  = 0, && i \in V \setminus (Z \cup T),  \\ 
& v(i) = 1, && i \in Z_m \cup T, \\
& v(i) = 0, && i \in Z_\ell, \ \ell \neq m. 
\end{align*} 
We first establish the following 
\begin{lemma}
\label{lem:singleSol}
$v_i \colon 2^{V\setminus Z} \to \mathbb R$ is a monotone submodular set function.
\end{lemma}
\begin{proof}
Fix an arbitrary subset $T \neq \varnothing$. Let $Z' = \cup_{\ell\neq m} Z_\ell$ and let $A = Z_m \cup T$. Recall the hitting time for a subset of vertices $U$ by a random walk is defined as $\tau(U) = \inf\{t\geq 0: X_t \in U\}$. We have that $v_i(T) = \Prob{i}{\tau(A) < \tau(Z')}$ by the probabilistic interpretation in \cref{s:probInterp}.

To prove that it is monotone note that if $T_1 = T$ and $T_2 = T\cup R$ we have 
\[
v_i(T_1) = \Prob{i}{\tau(A) < \tau(Z')} \leq \Prob{i}{\tau(A\cup R) < \tau(Z')} = v_i(T_2),
\]
where the inequality follows from the fact that $\tau(A\cup R) \leq \tau(A)$.

Now to prove submodularity consider the differences
\begin{align}
v_i(T\cup \{x\}) - v_i(T) 
&= \Prob{i}{\tau(A\cup\{x\}) < \tau(Z')} - \Prob{i}{\tau(A) < \tau(Z')} \notag \\
&= \Prob{i}{\tau({A\cup\{x\})} < \tau(Z') \leq \tau(A)} 
\label{eq:Ax}
\intertext{and}
v_i(T\cup \{x,y\} ) - v_i(T\cup \{y\})
&= \Prob{i}{\tau(A\cup\{x,y\}) < \tau(Z')} - \Prob{i}{\tau(A\cup\{y\}) < \tau(Z')} \notag \\
&= \Prob{i}{\tau(A\cup\{x,y\}) < \tau(Z') \leq \tau(A\cup\{y\})}.
\label{eq:Axy}
\end{align}
The second equality in \eqref{eq:Ax} follows\footnote{One can also show that it follows from the law of total probability:
\begin{align*}
\Prob{i}{\tau(A\cup x) < \tau(Z')}
  &  = \Prob{i}{\tau(A\cup x) < \tau(Z') | \tau(A) < \tau(Z')} \Prob{i}{\tau(A) < \tau(Z')} \\
   &\quad + \Prob{i}{\tau(A\cup x) < \tau(Z') | \tau(A) \geq   \tau(Z')} \Prob{i}{\tau(A) \geq \tau(Z')} \\
   &  = \Prob{i}{\tau(A) < \tau(Z')} + \Prob{i}{\tau(A\cup x) < \tau(Z') \leq \tau(A)},
\end{align*}
where we use the fact that $\Prob{i}{\tau(A\cup x) < \tau(Z') | \tau(A) < \tau(Z')} = 1.$} from the fact that the events 
\begin{align*}
E  = \big\{ \tau(A) < \tau(Z') \big\} \subset 
     \big\{ \tau(A\cup\{x\}) < \tau(Z') \big\} = F
\end{align*}
since if you hit $A$ before $Z'$ then you also hit $A\cup\{x\}$. Therefore,
\begin{align*}
\Prob{i}{\tau({A\cup\{x\})} < \tau(Z') \leq \tau(A)}  = \Prob{i}{F\cap E^c} = \Prob{i}{F} - \Prob{i}{E}. 
\end{align*}

The event $\big\{ \tau(A\cup\{x,y\}) < \tau(Z') < \tau(A\cup\{y\}) \big\}$ is equivalent to the random walk hitting $x$ before $Z'$ while avoiding $A\cup\{y\}$, this is contained in the event $\big\{ \tau(A\cup\{x\}) < \tau(Z') \leq \tau(A) \big\}$ which is equivalent to hitting $x$ before $Z'$ while avoiding just $A$. Therefore \eqref{eq:Axy} is smaller than \eqref{eq:Ax}. That is 
\[
v_i(T\cup \{x\}) - v_i(T) \geq v_i(T\cup \{x,y\} ) - v_i(T\cup \{y\}). 
\]
\end{proof}

\begin{proposition} \label{t:submodular}
The function $F_m \colon 2^{V\setminus Z} \to \mathbb R$ is a monotone and submodular set function satisfying 
\begin{itemize}
\item  $F_m(V\setminus Z) = \frac{|V \setminus Z| + |Z_m|}{|V|} \leq 1$. 
\item $F_m(\varnothing) = \frac{1}{|V|} \| v_\varnothing \|_1 \geq 0$  where $v_\varnothing$ satisfies 
\begin{align*}
& L v(i)  = 0, && i \in V \setminus Z ,\\ 
& v(i) = 1, && i \in Z_m , \\
& v(i) = 0, && i \in Z_\ell, \ \ell \neq m  .
\end{align*}
\end{itemize}
\end{proposition}
\begin{proof}
This follows from \cref{lem:singleSol} and $\mathcal{I}_m(T) = F_m(v_T) = \frac{1}{|V|} \sum_{i\in V} (v_T)_i$.
\end{proof}

\begin{corollary} \label{c:Approximation}
A greedy algorithm gives a $1-1/e$ approximation to the maximization of $F_m$ over sets of a given size.
\end{corollary}
\begin{proof}
The problem of maximizing a monotone submodular function $f$ subject to a cardinality constraint admits a $1-1/e$ approximation algorithm. The approximation is given by a greedy method, which starts with the empty set $S_0$, and in step $i$ adds the element of maximum increase to obtain $S_i$. Nemhauser et al.~\cite{nemhauser1978submodular} proved that $f(S_k) \geq (1-1/e) \max_{|S|\leq k} f(S)$.
\end{proof}

\section{A relaxation strategy for the targeting influence problem} \label{sec:Relax}
In \cref{s:authority}, we formulated a model for the targeting influence problem \eqref{e:OptProb}. This problem asks how an authority for opinion $m \in [k]$ would target network members in $V\setminus Z$ to maximize the proportion of influence that opinion $m$ has on the network. 
In \cref{s:binary}, we showed that this problem can be recast as a grouped opponent targeting influence problem \eqref{e:OptProb2}, where the authority doesn't have to consider all other opinions/ideas, but can lump them into a single group. In \cref{sec:HandA}, we showed that there exists an $\varepsilon > 0$, such that a $(1-\varepsilon)$-approximation is NP-hard, but that the problem can be  $(1-1/e)$-approximated with a polynomial-time algorithm. 
In this section we identify a convex relaxation of the NP-hard targeting influence problem.

For fixed $T \subset V\setminus Z$,  $\varepsilon>0$\creflastconjunction $m\in[k]$ 
we introduce the  problem 
\begin{subequations}
\label{e:LabelPropSysRelax}
\begin{align}
& L u(i) + \varepsilon^{-1} e_T(i) (u - e_m)(i) = 0,  && i \in V \setminus Z, \\ 
 &u(i) = e_\ell, && i \in Z_\ell, \ \ell \in [k]. 
\end{align}
\end{subequations}
Here, we have used the notation $e_T = \begin{cases} 1 & i \in T \\ 0 & i \notin T\end{cases}$ to be the characteristic vector for the set $T \subset V\setminus Z$.  Throughout, we will let $e := e_V$ to mean the vector of all $1$'s.

\begin{lemma} \label{l:RelaxConv}
As $\varepsilon \downarrow 0$, the solution, $ u_\varepsilon$ of \eqref{e:LabelPropSysRelax} converges to the solution $u$ of the graph Laplace problem 
\cref{e:OptProbB,e:OptProbC,e:OptProbD}; in particular, $u_\varepsilon(i) \to e_m$ for $i\in T$. 
\end{lemma}
\cref{l:RelaxConv} can be proved by considering the difference between solutions $u_\varepsilon$ and $u$, and applying a discrete maximum principle on $V \setminus Z$.

For $\varepsilon > 0$, a relaxed formulation of the targeting influence problem \eqref{e:OptProb} is then given by 
\begin{subequations}
\label{e:OptProbRelax}
\begin{align}
\max_{\phi\colon V \to \mathbb R}  \ & \mathcal{I}_m(u) \\
\label{e:OptProbRelaxA}
\textrm{s.t.} \ &  L u(i)  + \varepsilon^{-1} \phi \odot (u - e_m)(i) = 0,  && i \in V \setminus Z, \\ 
\label{e:OptProbRelaxB}
 & u(i) = e_\ell, && i \in Z_\ell, \ \ell \in [k], \\ 
\label{e:OptProbRelaxC}
& \phi \geq 0, \ 
\phi\!\mid_Z = 0,  \ 
\sum_{i\in V} \phi(i) = 1. 
\end{align}
\end{subequations}
Here, we have replaced the characteristic function $e_T\colon V \to \{0,1\}^n$ for $T \subset V \setminus Z$, $|T| = t$ by the function $\phi \colon V \to \mathbb R_+$, that is constrained to be supported on the set $V \setminus Z$ and satisfies  $\sum_{i\in V} \phi(i) = 1$.

\begin{remark}
In \eqref{e:OptProbRelaxC}, we have constrained $\phi$  to satisfy $\|\phi\|_{\ell^1(V)} =1$. 
Comparing to \eqref{e:OptProbE}, one might think that this constraint should be $\|\phi\|_{\ell^1(V)} =t$. However, this problem only depends on the ratio $\varepsilon/t$, so we can set $t=1$ for convenience. This invariance signifies a shortcoming of the relaxed model; we do not expect that $\phi$ is supported on exactly $t$ vertices. 
\end{remark}

Finally, as described in \cref{s:binary}, we can reduce this to an equivalent, us-vs-them binary problem. Writing $v = u_m$ and again abusing notation by writing 
$\mathcal{I}_m(v) = \frac{1}{|V|} \|v\|_{\ell^1(V)}$, we have 
\begin{subequations}
\label{e:OptProb2Relax}
\begin{align}
\label{e:OptProb2RelaxA}
\max_{\phi\colon V \to \mathbb R}  \ & \mathcal{I}_m(v), \\
\label{e:OptProb2RelaxB}
\textrm{s.t.} \ &  L v(i)  + \varepsilon^{-1} \phi \odot (v - e)(i) = 0,  && i \in V \setminus Z, \\ 
\label{e:OptProb2RelaxC}
 & v(i) = 1, && i \in Z_m, \\
 \label{e:OptProb2RelaxD}
 & v(i) = 0, && i \in Z_\ell, \ \ell\neq m, \\ 
\label{e:OptProb2RelaxE}
& \phi \geq 0, \ 
\phi\!\mid_Z = 0,  \ 
\sum_{i\in V} \phi(i) = 1. 
\end{align}
\end{subequations}
Equation \eqref{e:OptProb2Relax} is a relaxed version of the targeting influence problem that we will solve computationally. 
If $\phi^\star$ is the solution to \eqref{e:OptProb2Relax}, we  obtain a set $T \subset V \setminus Z$ by choosing the largest $t$ values of $\phi^\star$, arbitrarily breaking ties if necessary.

The following theorem shows that the objective function is strictly concave in $\phi$ and gives expressions for the gradient and Hessian. 
To state the theorem, we use the following notation. The function $\phi$ is constrained to have support on $Z^c = V\setminus Z$, so it is useful to partition the variables according to the sets $\{Z,Z^c\}$. 
For $v \in \mathbb R^{|V|}$, 
we write $v_c = v|_{Z^c}$ to denote the entries corresponding to  $Z^c$. 
For $L \in \mathbb R^{|V| \times |V|}$, 
we write $L_{c,c}$ to denote the principle submatrix corresponding to 
$Z^c \times Z^c$.

\begin{theorem}
\label{t:Convexity}
Let $G=(V,E)$ be a strongly connected, directed graph. 
For $\phi\colon V \to \mathbb R$ satisfying the constraints in \eqref{e:OptProb2RelaxE}, let $v_\phi$ denote the solution to equation  \eqref{e:OptProb2RelaxB} satisfying the Dirichlet boundary conditions \eqref{e:OptProb2RelaxC} and \eqref{e:OptProb2RelaxD}.
The map $\phi \mapsto \mathcal{I}_m(v_\phi)$ is strictly concave and twice differentiable. The gradient is given by 
\begin{equation}
\label{e:grad_dir}
\nabla_\phi \mathcal{I}_m \mid_c = \frac{1}{|V|} (e-v_c) \odot w_c
\end{equation}
where $w_c \colon V\setminus Z \to \mathbb R$ is defined by
$w_c := \varepsilon^{-1}\left[ L_{c,c} + \varepsilon^{-1} \phi_c \right]^{-T} e$.
The Hessian is given by
\begin{equation}
\label{e:Hessian_dir}
\nabla^2_\phi \mathcal{I}_m \mid_{c,c} \ = \  -  \frac{2}{\varepsilon |V|}  \  \textrm{sym}\left(  \left[ L_{c,c} + \varepsilon^{-1} \phi_c \right]^{-1}  \odot [ w_c \otimes (e-v_c) ] \right), 
\end{equation}
where $\textrm{sym}(A) = (A + A^T)/2$. 
\end{theorem}

\begin{proof}
 We first compute the gradient and Hessian. We additionally need the notation that the submatrix of $L$ corresponding to $V\setminus Z \times Z_m$ is denoted by $L_{c,m}$. 
We can then rewrite \eqref{e:OptProb2RelaxB}-\eqref{e:OptProb2RelaxD} for the unknown $v_c$ as  
\begin{equation}
\label{e:partEq}
L_{c,c} v_c + L_{c,m} e + \varepsilon^{-1} \phi_c \odot \left( v_c - e \right) =0. 
\end{equation}
We can write the solution $v_c$ in terms of the Schur complement 
\begin{equation}
    \label{e:SchurSol}
v_c =  [L_{c,c} + \varepsilon^{-1}\phi_c]^{-1}  \left( \varepsilon^{-1} \phi_c  - L_{c,m}e \right).
\end{equation}

Differentiating \eqref{e:partEq} with respect to $\phi_j$, with $j \in Z^c$ we obtain 
$$
\left( L_{c,c} +  \varepsilon^{-1} \phi_c \right) \frac{\partial v_c}{\partial \phi_j} =  
 \varepsilon^{-1} e_j \odot \left( e - v_c \right) 
$$
giving that 
$$
\frac{\partial v_c}{\partial \phi_j} =  
 \varepsilon^{-1} [\left( L_{c,c} +  \varepsilon^{-1} \phi_c \right)]^{-1} \left( e_j \odot \left( e - v_c \right) \right). 
$$
Thus, 
\begin{align*}
\frac{\partial \mathcal{I}_m}{\partial \phi_j}
&= \frac{1}{|V|} \sum_{i \in Z^c} \frac{\partial v}{\partial \phi_j} (i) 
= \frac{1}{|V|} \left\langle  e, \frac{\partial v_c}{\partial \phi_j} \right\rangle \\
&= \varepsilon^{-1} \frac{1}{|V|} \left\langle e,  \left[ L_{c,c} + \varepsilon^{-1} \phi_c \right]^{-1} \left( e_j \odot \left( e - v_c  \right) \right) \right\rangle\\ 
&= \varepsilon^{-1} \frac{1}{|V|} \left\langle \left( e - v_c  \right) \odot \left[ L_{c,c} + \varepsilon^{-1} \phi_c \right]^{-T} e  ,  e_j \right\rangle \\
&= \frac{1}{|V|} \left\langle \left( e - v_c  \right) \odot w_c, e_j \right\rangle. 
\end{align*}
This gives that 
$\nabla_{\phi} J \mid_c =  \frac{1}{|V|} \left( e - v_c \right) \odot w_c$, 
as desired. 

Differentiating \eqref{e:partEq} a second time with respect to $\phi_k$ with $k \in Z^c$ yields 
$$
\left( L_{c,c} +  \varepsilon^{-1} \phi_c \right) \frac{\partial^2 v_c}{\partial \phi_j \partial \phi_k} =  
- \varepsilon^{-1} \left( e_k \odot \frac{\partial v_c}{\partial \phi_j} + e_j \odot \frac{\partial v_c}{\partial \phi_k} \right). 
$$
Solving for $\frac{\partial^2 v_c}{\partial \phi_j \partial \phi_k}$ and using the definition of $\mathcal{I}_m$, we have that 
\begin{align*}
\frac{\partial^2 \mathcal{I}_m}{\partial \phi_j \partial \phi_k} 
&= \frac{1}{|V|} \left\langle e, \frac{\partial^2 v_c}{\partial \phi_j \partial \phi_k} \right\rangle \\
&= - \frac{1}{|V|} \left\langle w_c, e_k \odot \frac{\partial v_c}{\partial \phi_j} + e_j \odot \frac{\partial v_c}{\partial \phi_k} \right\rangle \\ 
&= - \frac{1}{|V|} \left( \langle w_c, e_k\rangle \left\langle e_k, \frac{\partial v_c}{\partial \phi_j}  \right\rangle + \langle w_c, e_j \rangle  \left\langle e_j,  \frac{\partial v_c}{\partial \phi_k} \right\rangle \right),
\end{align*}
where we have used the property that $\langle a,e_k \odot b\rangle = \langle a,e_k \rangle \langle b, e_k\rangle$. 
We then use this property again to compute 
\begin{align*}
\left\langle e_k,  \frac{\partial v_c}{\partial \phi_j} \right\rangle 
& = \varepsilon^{-1} \left\langle \left[ L_{c,c} + \varepsilon^{-1} \phi_c \right]^{-T} e_k, e_j \odot (e-v_c) \right\rangle \\
& = \varepsilon^{-1} \left\langle \Phi e_k, e_j   \right\rangle \langle e_j, e-v_c \rangle,
\end{align*}
where we have written $\Phi :=  \left[ L_{c,c} + \varepsilon^{-T} \phi_c \right]^{-1}$.
Substituting this above, we have 
\begin{align*}
\frac{\partial^2 \mathcal{I}_m}{\partial \phi_j \partial \phi_k} 
&= -  \frac{1}{\varepsilon |V|} \left( \langle w_c, e_k\rangle \left\langle \Phi e_k, e_j   \right\rangle \langle e_j, e-v_c \rangle + \langle w_c, e_j \rangle  \left\langle \Phi e_j, e_k   \right\rangle \langle e_k, e-v_c \rangle\right) \\
&= \langle e_j, H e_k \rangle,
\end{align*}
where 
$H = -  \frac{1}{\varepsilon |V|} \left( \Phi \odot  [(e-v_c) \otimes w_c ]+ \Phi^T \odot [ w_c \otimes (e-v_c)  ] \right) $ as desired. 

We now prove strict concavity. Our argument follows the proof of  in \cite[Prop. 2.4]{Boyd2023}, using the Neumann series (i.e., $(I-T)^{-1} = \sum_{k=0}^\infty T^k$) in the Schur complement expression of the solution \eqref{e:SchurSol} and investigating convexity of each term. 
Writing 
$x = d_{cc} + \varepsilon^{-1} \phi$
and $X = \textrm{diag}(x) = D_{c,c} + \varepsilon^{-1} {\rm diag} (\phi)$ and noting $L_{c,m} = - A_{c,m}$, we have
\begin{subequations}
\begin{align} 
|V| \mathcal{I}_m (v_\phi) 
&= \langle e,v \rangle \\ 
&=  \left\langle e, \  \left( X - A_{c,c} \right)^{-1} ( \varepsilon^{-1} \phi- L_{c,m}e)  \right\rangle \\
&=  \left\langle e, \  \left( I - X^{-1} A_{c,c} \right)^{-1} X^{-1} ( \varepsilon^{-1} \phi+ A_{c,m}e)  \right\rangle \\
&= \sum_{k=0}^{\infty} \left\langle e, \   \left(X^{-1} A_{c,c} \right)^k X^{-1} ( \varepsilon^{-1} \phi + A_{c,m} e) \right\rangle.
\label{eq:J1}
\end{align}
\end{subequations}

Let us reframe the problem in terms of the variable $x$.  Then,~\eqref{eq:J1} becomes
\begin{align*}
   |V|  E(x) & =\left\langle e, \  \left( I - X^{-1} A_{c,c} \right)^{-1} X^{-1} ( x - d_{cc} + A_{c,m}e)  \right\rangle \\
    & = \left\langle e, \sum_{k=0}^\infty (X^{-1} A_{c,c})^k (e -X^{-1}(d_{cc} + A_{c,m} e)) \right\rangle \\
    & = \left\langle e,  \sum_{k=0}^\infty (X^{-1} A_{c,c})^ke -\sum_{k=0}^\infty (X^{-1} A_{c,c})^kX^{-1}(d_{cc} + A_{c,m} e) \right\rangle \\
    & = \langle e, e \rangle + \left\langle  e, \sum_{k=1}^\infty (X^{-1} A_{c,c})^ke -\sum_{k=0}^\infty (X^{-1} A_{c,c})^kX^{-1}(d_{cc} + A_{c,m} e) \right\rangle \\
    & = |Z^c| + \left\langle  e, \sum_{k=0}^\infty (X^{-1} A_{c,c})^k X^{-1} A_{c,c} e -\sum_{k=0}^\infty (X^{-1} A_{c,c})^kX^{-1}(d_{cc} + A_{c,m} e) \right\rangle \\
    & = |Z^c| + \left\langle  e, \sum_{k=0}^\infty (X^{-1} A_{c,c})^k X^{-1} \left( A_{c,c} e - d_{cc} - A_{c,m} e \right) \right\rangle.
\end{align*}
Notice that $A_{c,c} e - (d_{c,c} + A_{c,m}) e$ is nonpositive because $A_{c,c}e$ is the degree of each node, counting only neighbors in $c$, whereas $d_{c,c}$ is the full degree. 

We first show concavity of $E(x)$. Since the sum of convex functions is convex, it is enough to show that 
\[e^T \left( X^{-1} A_{c,c} \right)^k X^{-1} f\]
is convex for each $ k > 0 $, for all nonnegative vectors $f$. 
For fixed $k$, and using the fact that $e$, $f$\creflastconjunction $A_{c,c}$ are nonnegative, the preceding term is a nonnegative linear combination of terms of the form
\[ F(x) = \prod_i x_i^{-\alpha_i}.\]
Therefore, we seek to show that $F$ is convex for any $\alpha = (\alpha_1,\cdots,\alpha_n)$ with each $\alpha_j \geq 0$. 
To prove this, we check whether the Hessian is semi-positive definite. Computing second derivatives gives
$$
\frac{\partial^2 F}{ \partial x_i \partial x_i} (x) = F(x) \alpha_i (\alpha_i + 1) x_i^{-2}
$$
and 
$$
\frac{\partial^2 F}{ \partial x_i \partial x_j} (x) = F(x) \alpha_i \alpha_j  x_i^{-1} x_j^{-1}.
$$
So the Hessian of $F$ is 
$$
F(x)  \left[ (\alpha \odot x^{-1}) (\alpha \odot x^{-1})^T + \textrm{diag}(\alpha \odot x^{-2})\right]. 
$$
The first term is positive semidefinite (PSD) since it is the outer product of a vector with itself and the second term is PSD since it is diagonal with positive entries. Thus, the Hessian of $F$ is PSD.

To prove positive definiteness (rather than semi-definiteness), recognize that the $k = 0$ term contributes a term to the Hessian of the form $DX^{-2}$, which is strictly positive definite on the domain in question. This proves strict convexity.
\end{proof}

\begin{remark}
In \cref{s:energyInterp}, we gave an 
energy interpretation in the case where the graph is undirected. 
For fixed $T \subset V\setminus Z$,  $\varepsilon>0$\creflastconjunction $m\in[k]$, define the relaxed energy 
$$
E_{m, \varepsilon} (u; T) = E(u) + \frac{1}{2 \varepsilon} \sum_{i \in T} | u(i) - e_m |^2 . 
$$
The minimizers of this energy satisfy \eqref{e:LabelPropSysRelax}. 
Furthermore, as $\varepsilon \downarrow 0$, the minimizer $u = u_\varepsilon$ will converge to the solution of the graph Laplace problem 
\cref{e:OptProbB,e:OptProbC,e:OptProbD}. In particular, $u_\varepsilon(i) \to e_m$ for $i\in T$.
\end{remark}

\subsection{Targeting influence on symmetric graphs} \label{sec:GrAut}

In this section, we consider the relaxation of the grouped opponent targeting influence problem \eqref{e:OptProb2Relax} in the case when the graph has a symmetry and prove that the optimization shares the symmetry. 

Recall that an \emph{automorphism} of a graph $(V,E)$ is a permutation $\pi \colon V \to V$ such that $(i,j) \in E$ if and only if $(\pi(i),\pi(j)) \in E$. An automorphism $\pi$ induces a linear transformation  $P_\pi$, on $\ell^2(V)$, defined by $(P_\pi(u)))_i = u_{\pi(i)}$, $u \in \ell^2(V)$, $i \in V$. The matrix $P_\pi \in \mathbb R^{|V|\times |V|}$ is a permutation matrix. 

\begin{lemma} 
Suppose that $\pi$ is an automorphism on $(V,E)$ and let $P_\pi$ be the induced linear transformation. 
Suppose $S\subset V$ is a permutation invariant set, i.e. $\pi S = S$.  
Let $f_i \in \ell^2(V\setminus S)$, $i=1,2$, and $f_3 \in \ell^2(S)$ and extend these functions by 0 to $V$ and abuse notation by denoting the extensions by $f_i \in \ell^2(V)$, $i=1,2,3$. 
Assume $f_i \geq 0$ and are invariant to $P_\pi$, $i=1,2,3$, i.e., $P_\pi f_i = f_i$. 
Consider the equation 
\begin{subequations}
\label{e:GenEq}
\begin{align}
Lv(i)+ f_1(i) v(i) & = f_2 (i), && i \in V\setminus S, \\
v(i) &= f_3(i), && i \in S. 
\end{align}
\end{subequations}
Then the solution  $v$ is also invariant to $P_\pi$.
\end{lemma}

\begin{proof}
Since $\pi$ is an automorphism, $P_\pi$ commutes with the adjacency matrix $A$, the degree matrix $D$, and hence the graph Laplacian $L = D-A$. Multiplying both equations in \eqref{e:GenEq} by $P_\pi$ and using commutativity, we have 
\begin{align*}
L (P_\pi v) (i)+ (P_\pi f_1)(i) (P_\pi v) (i) & = (P_\pi f_2) (i), && i \in V\setminus S, \\
(P_\pi v) (i) &= (P_\pi f_3)(i), && i \in S, 
\end{align*}
which shows that $P_\pi v$ is also a solution to \eqref{e:GenEq}.  By uniqueness of the solution to \eqref{e:GenEq}, we have that $v = P_\pi v$. 
\end{proof}

We similarly consider the relaxation of the grouped opponent  targeting influence problem \eqref{e:OptProb2Relax}.

\begin{proposition}
\label{p:GraphAuto}
Let $\pi$ be an automorphism on the strongly connected, directed graph $(V,E)$ with induced linear transformation $P_\pi$. 
Fix zealot sets $Z = \amalg_{\ell \in [k]} \ Z_\ell$. 
Suppose that $Z_m$ and $\cup_{\ell \neq m} Z_\ell$ are permutation invariant sets. 
Fix $\varepsilon > 0$ and consider \eqref{e:OptProb2Relax}. The unique solution $\phi^\star$ is invariant to $P_\pi$.
\end{proposition}

\begin{proof}
Let $\phi^\star$ be the solution  to \eqref{e:OptProb2Relax} and let  $v^\star$ be the corresponding harmonic function. 
Since $P_\pi$ is finite, there exists an integer $S$, so that $P_\pi^S = I$. We consider the test function 
$$
\tilde \phi = \frac{1}{S} \sum_{s \in [S]} P_\pi^s \phi^*, 
$$
which is invariant to $P_\pi$. For each $s\in{S}$, denote the solution to the problem with potential $P_\pi^s \phi^\star$ by $v_s$. We then use the linearity of the constraints to argue that $\tilde v = \frac{1}{S} \sum_{s \in [S]} v_s = \frac{1}{S} \sum_{s \in [S]} P_\pi^s v^\star$ is a solution for $\tilde \phi$. The strict concavity of $\mathcal{I}_m$ (\cref{t:Convexity}) and Jensen's inequality then gives that 
$$
\mathcal{I}_m(\tilde v) > \frac{1}{S} \sum_{s \in [S]} \mathcal{I}_m(v_s) 
=  \frac{1}{S} \sum_{s \in [S]} \mathcal{I}_m(v^\star)  
= \mathcal{I}_m (v^\star) 
$$
unless $\tilde\phi=\phi^*$. Here we use that the $\ell^1$-norm is invariant to permutation of the argument. 
This strict inequality contradicts the optimality of $\phi^\star$. 
\end{proof}

\section{Computational methods and experiments} \label{sec:Num}
In this section, we perform several numerical experiments to investigate and illustrate the targeting influence problem
\eqref{e:OptProb2}. We first describe our implementation of the two computational approaches to this NP-hard problem; see \cref{alg:GreedyAlg,alg:RelaxAlg}. Both algorithms were implemented in MATLAB and run on a laptop computer.

\begin{algorithm}[t]
\caption{Greedy algorithm for solving grouped opponent targeting influence problem
\eqref{e:OptProb2}}
\label{alg:GreedyAlg}
\begin{algorithmic}
\State {\bf Input}: directed graph, authority index $m \in [k]$, initial zealot sets $Z = \amalg_{\ell \in [k]} Z_\ell$\creflastconjunction an augmented zealot set size $t \in \mathbb N\setminus 0$. 
\State {\bf Output}: an updated zealot set $Z_m \cup T$ of size $|Z_m \cup T| = |Z_m| + t$. 
\State Set $T = \varnothing$. 
\For {$s \in [t]$} 
\For{$i \in V\setminus (Z\cup T)$} 
\State Set $\tilde T = T \cup \{i\}$ and solve 
\begin{align*}
 \ & L v(i)  = 0, && i \in V \setminus (Z \cup \tilde T), \\ 
& v(i) = 1, && i \in Z_m \cup \tilde T,  \\
& v(i) = 0, && i \in Z_\ell, \ \ell \neq m .
\end{align*}

\State Compute the energy 
$$
E(i) = \mathcal{I}_m(v) = \frac{1}{|V|} \| v \|_{\ell^1(V)}.
$$
\EndFor
\State  Compute
$
j=\mathrm{argmax}_{i \in V\setminus (Z\cup T)}  \ E(i), 
$
breaking ties randomly if needed.
\State Set $T = T \cup \{j\}$. 
\EndFor
\end{algorithmic}
\end{algorithm}

\begin{algorithm}[t]
\caption{Algorithm for approximately solving  the  grouped opponent targeting influence problem \eqref{e:OptProb2} utilizing the relaxation
\eqref{e:OptProb2Relax}.}
\label{alg:RelaxAlg}
\begin{algorithmic}
\State {\bf Input}: directed graph, authority index $m \in [k]$, 
initial zealot sets $Z = \amalg_{\ell \in [k]} Z_\ell$\creflastconjunction 
an augmented zealot set size $t \in \mathbb N\setminus 0$. 
\State {\bf Output}: an updated zealot set $Z_m \cup T$ for of size $|Z_m \cup T| = |Z_m| + t$. 
\State Set $T = \varnothing$. 
\For {$s \in [t]$} 
\State Solve the convex optimization problem 
\begin{align*}
\max_{\phi\colon V \to \mathbb R}  \ & \mathcal{I}_m(v) \\
\textrm{s.t.} \ &  L v(i)  + \varepsilon^{-1} \phi \odot (v - e)(i) = 0,  && i \in V \setminus (Z\cup T) , \\ 
 & v(i) = 1, && i \in Z_m \cup T, \\
 & v(i) = 0, && i \in Z_\ell, \ \ell\neq m, \\ 
& \phi \geq 0, \ 
\phi\!\mid_{Z\cup T} = 0,  \ 
\sum_{i\in V} \phi(i) = 1. 
\end{align*}

\State  Compute
$
j=\mathrm{argmax}_{i \in V\setminus (Z\cup T)} \  \phi(i), 
$
breaking ties randomly if needed.
\State Set $T = T \cup \{j\}$. 
\EndFor
\end{algorithmic}
\end{algorithm}

\subsection{Computational Methods}
We first implement a greedy algorithm to solve \eqref{e:OptProb2}; see \cref{alg:GreedyAlg}.  
For each opinion $m$, we sweep through all non-zealot nodes, $i \in V \setminus Z$, and solve the version of \eqref{e:OptProb2} with Dirichlet boundary conditions appropriately chosen. For the relatively small size graphs considered here, we simply use the MATLAB \texttt{backslash} operator to solve the Schur complement problem.  After evaluating the energy for each non-zealot node, $E(i)$, we find the vertex which gives the largest energy and add it to the zealot set $m$.  Since each iteration in $s$ requires $O(n)$ linear solves each costing $O(n^3)$, the run time of this algorithm is $O(t n^4)$.

In \cref{alg:RelaxAlg}, we implement an algorithm based on the relaxation of the problem described in \cref{sec:Relax}. Given existing zealot sets, we use a quasi-Newton BGFS method in MATLAB via \texttt{fmincon} in order to approximate the unique optimal solution to \eqref{e:OptProb2Relax}.  We solved \eqref{e:OptProb2RelaxA} using the MATLAB least squares function \texttt{lsqr}. Here, we dramatically speed up the numerical optimizations by including the gradient as computed in Theorem \ref{t:Convexity}, which we verified by using the \texttt{Gradient Check} feature in the optimization solver. 
Each iteration in $s$ requires the solution of a single optimization problem. The time complexity of computing the gradient is $O(n^3)$  and we perform $O(1)$ iterations to approximate the solution to the optimization problem. The run time for \cref{alg:RelaxAlg} is then $O(t n^3)$, an $O(n)$ reduction over \cref{alg:GreedyAlg}. 
A choice of $\varepsilon>0$ must be made and we discuss this choice in the examples below. Generally, we observe that for very small $\varepsilon$ the algorithm prioritizes proximity to the existing zealot nodes, while for $\varepsilon \sim \| L \|_{\rm Frobenius}^{-1}$, the graph dynamics play a more significant role.

In both algorithms, the initialization of zealot nodes can be done by, e.g., targeting a specific region of the graph, by random initial assignment, or by using a dynamic programming approach to identify the best possible first move for one of the label sets.

\subsection{Computational experiments} 
\label{s:CompExp}
In the general us-vs-them targeting influence problem  \eqref{e:OptProb}  and its reduction to a group opponent problem \eqref{e:OptProb2}, one may consider the case where multiple authorities $m\in[k]$ compete for non-zealot individuals in the network at the same time. 
For simplicity, in each experiment below, we consider a game between two opinion authorities to capture the largest proportion of influence. Each authority alternatively selects one vertex to convert to their opinion.

\begin{figure}[t]
    \centering
\includegraphics[width=.49\textwidth]{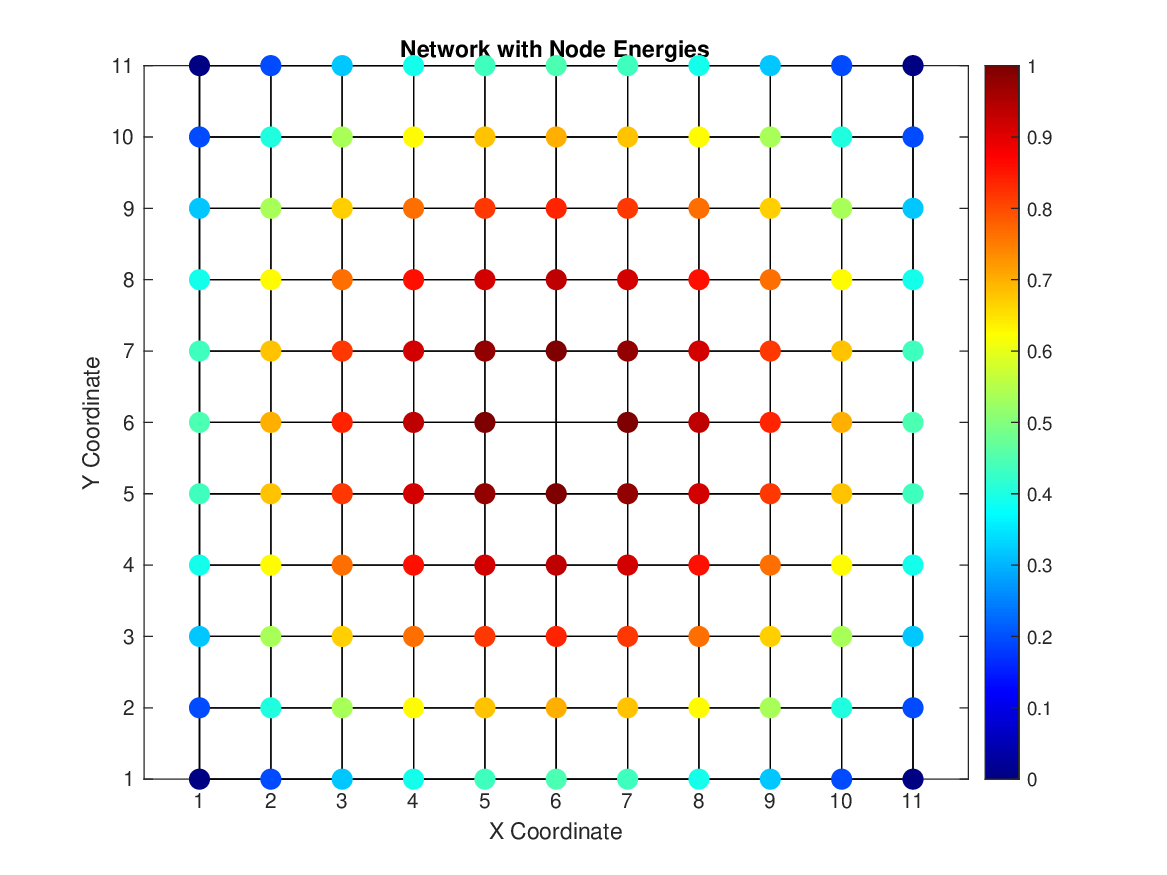} 
\includegraphics[width=.49\textwidth]{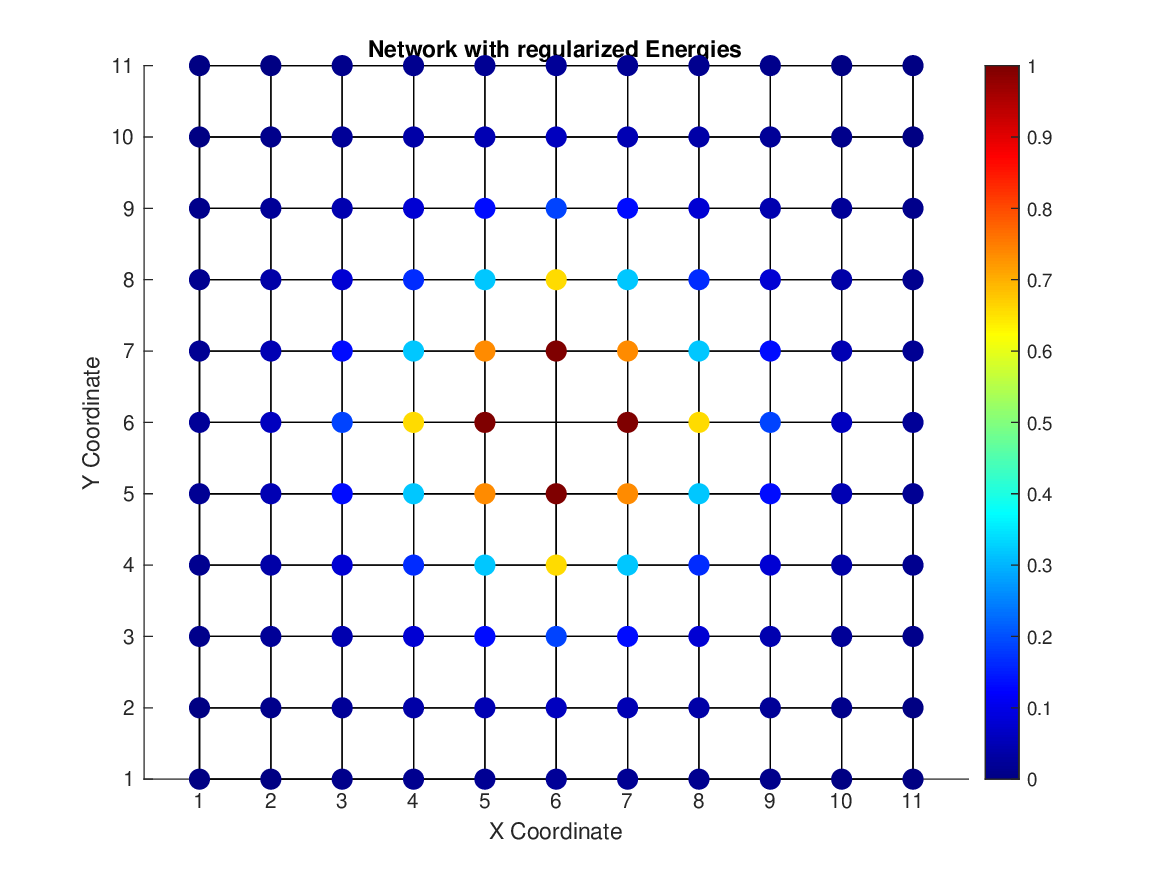} 
    \caption{On an $11\times 11$ square grid graph, we consider the case where first authority has chosen the center node and the second authority is choosing a vertex. {\bf (left)} We plot the measure of influence $\tilde{\mathcal{I}}_2$ for the second authority for each available node. 
 {\bf (right)} We plot the $\tilde \phi$ value obtained as the solution to \cref{alg:RelaxAlg} with $\varepsilon = .15$. The  $\tilde \phi$ values respect the underlying symmetry and correctly predict that the optimal single-node choice is to pick one of the 4 nearest neighbors. See \cref{s:CompExp1}. }
    \label{fig:square}
\end{figure}

\subsubsection{Square grid} \label{s:CompExp1}
In a first experiment, we consider an $11\times 11$ square grid graph. The first opinion authority has chosen $Z_1$ to be the center node. In \cref{fig:square}, we compare the energy 
that the second opinion authority would realize choosing each node vs.\ the prediction via the relaxed problem \eqref{e:OptProb2Relax}. 
On the left of \cref{fig:square}, we plot for each $i \in V \setminus Z_1$, a normalized  energy $$
 \tilde{\mathcal{I}}_2(i) := \frac{\mathcal{I}_2(i) - \min_j \mathcal{I}_2(j) } {\max_j \mathcal{I}_2(j) - \min_j \mathcal{I}_2(j)}, 
 \qquad 
 i \in V \setminus Z_1.
$$   
where $v$ satisfies the boundary condition $v=0$ on $Z_1$ and $v=1$ on  $Z_2 = \{i\}$.  
On the right, we solve \eqref{e:OptProb2Relax} with $\varepsilon = .15$ and plot the normalized vertex function, normalized similar to above, 
$$
 \tilde \phi(i) := \frac{\phi(i) - \min_j \phi(j) } {\max_j \phi(j) - \min_j \phi(j)}, 
 \qquad 
 i \in V \setminus Z_1.
$$
As predicted by \cref{p:GraphAuto}, this normalized vertex function $\tilde \phi$ has $4$-fold symmetry. 
From the left plot, we see that it is advantageous to choose a neighboring node of $Z_1$ and the plot of $\tilde \phi$ in the right panel is a good predictor of this behavior.

\begin{figure}[t]
    \centering
\includegraphics[width=.49\textwidth]{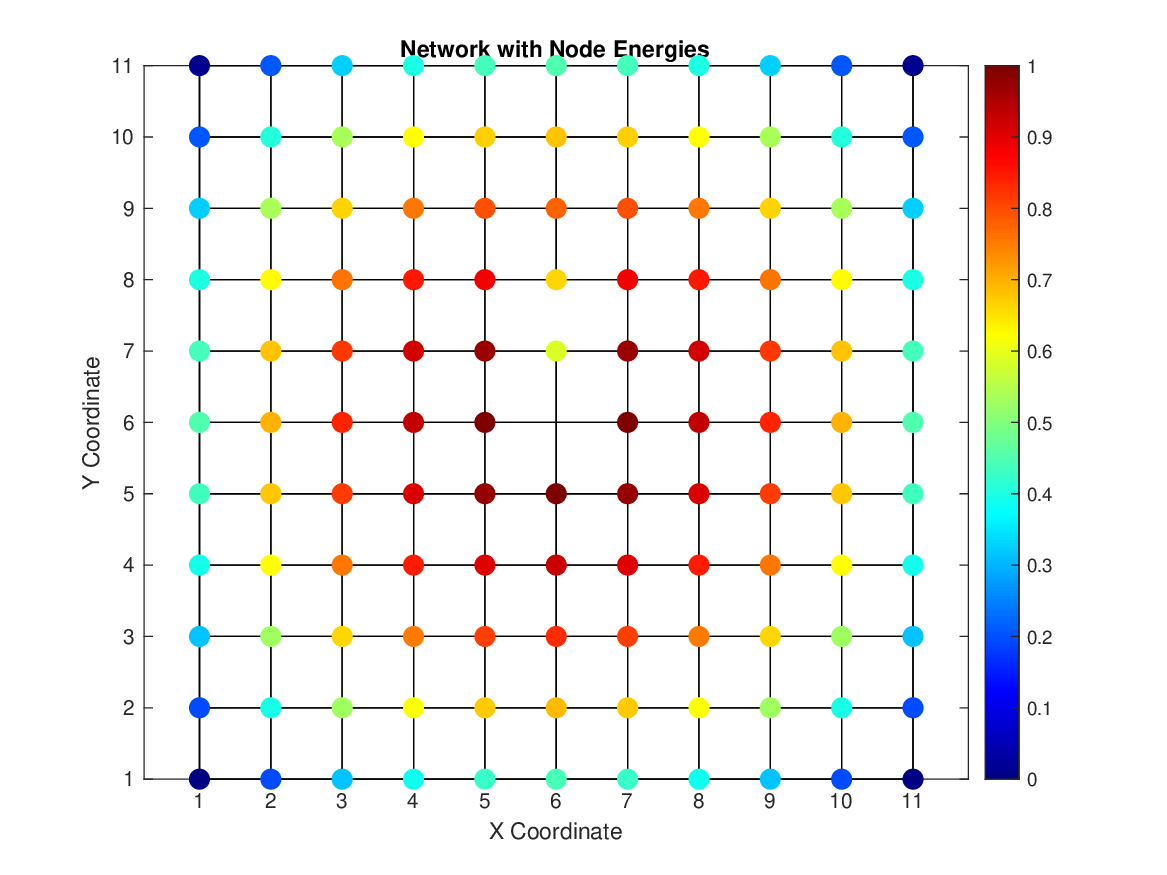}
\includegraphics[width=.49\textwidth]{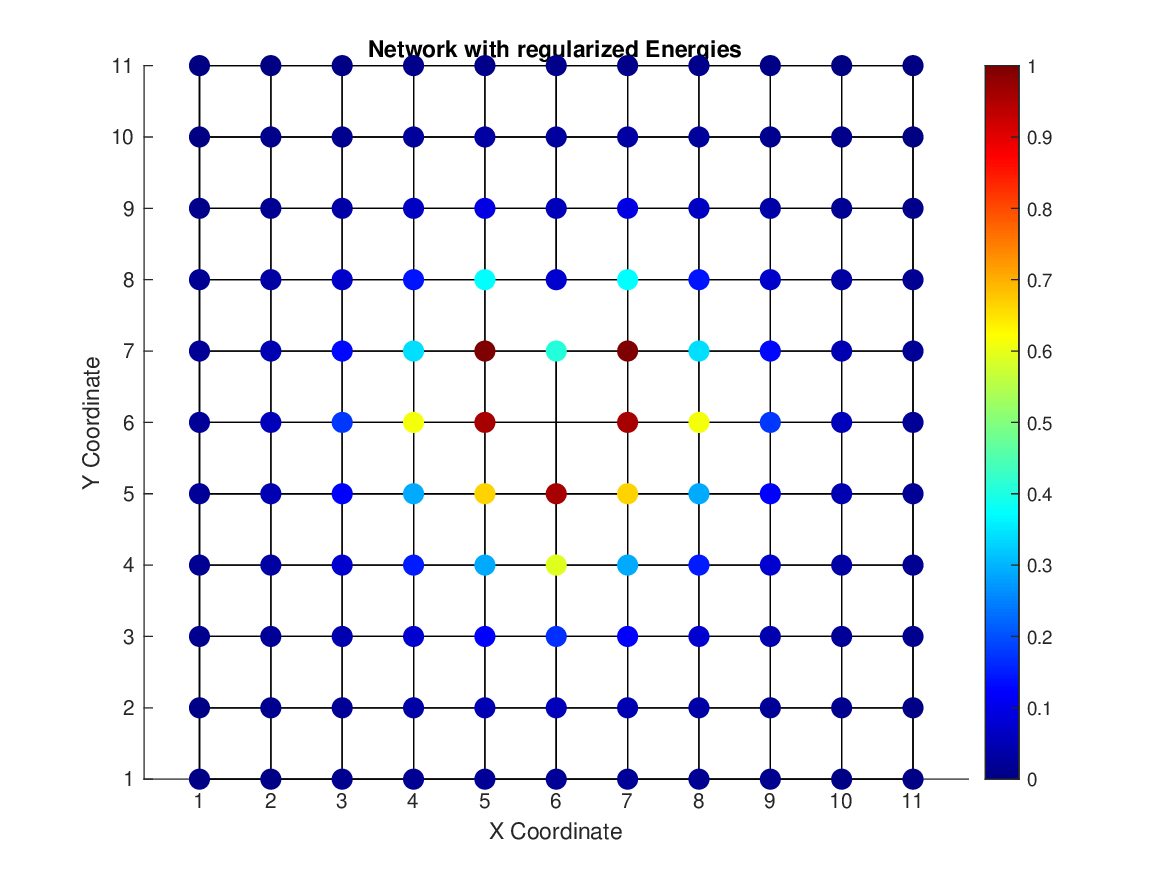}
    \caption{We again consider an $11\times 11$ square grid graph, but now with a single edge missing. The left and right panels are as in \cref{fig:square}. See \cref{s:CompExp2}.}
    \label{fig:squaredefect}
\end{figure}

\subsubsection{Square grid with a defect} \label{s:CompExp2}
We repeat the first experiment (see \cref{s:CompExp1}), but for an $11\times 11$ square grid graph with a defect -- the edge connecting vertex $(6,7)$ with vertex $(6,8)$ has been removed. The results are plotted in \cref{fig:squaredefect}. 
As predicted by \cref{p:GraphAuto}, the vertex function $\tilde \phi$ has a horizontal reflection symmetry, but the vertical reflection symmetry has now been broken. 
From the left plot, we see that it is advantageous to choose a neighboring node of $Z_1$ that is below $Z_1$. The neighboring node above $Z_1$ is less advantageous; because of the missing edge, there are fewer edges for the opinion to propagate from that node. Again, the function from the relaxed problem plotted in the right does a good job predicting this asymmetry, and also highlighting that the nodes diagonal up to the right/left are preferable; while the specific rank order of these diagonal nodes relative to the nearest neighbors down, left and right differ slightly between the two panels, the impacts are very similar from each of these nodes. 

\begin{figure}
    \centering
\includegraphics[width=.49\textwidth]{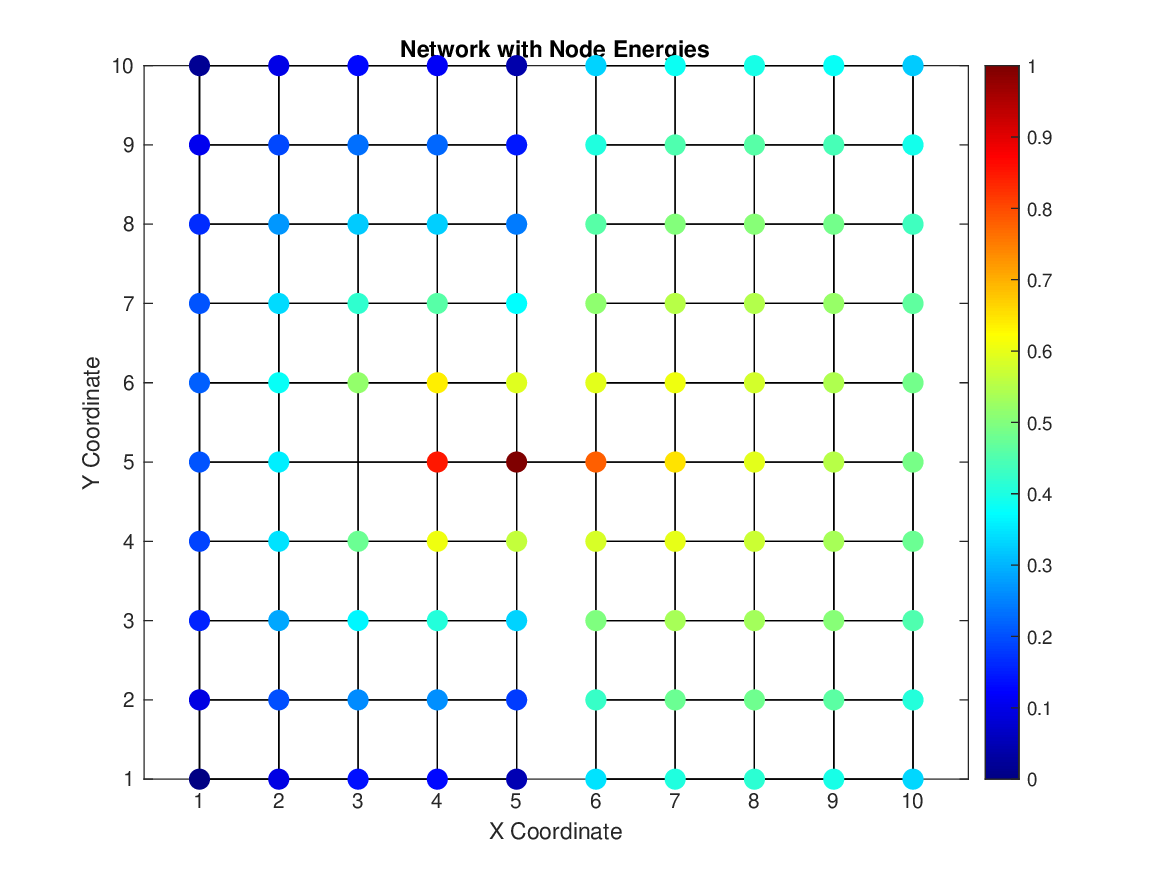}  
\includegraphics[width=.49\textwidth]{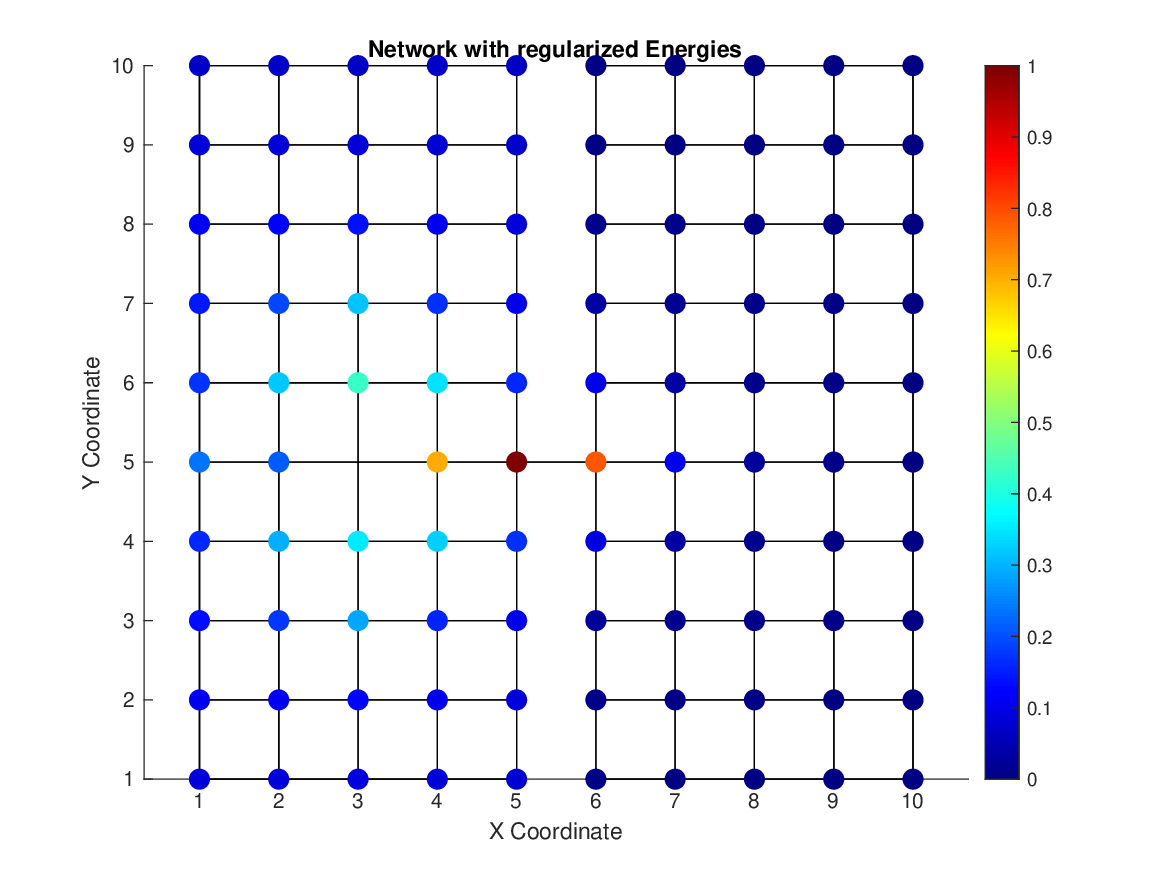} \\ 
\includegraphics[width=.49\textwidth]{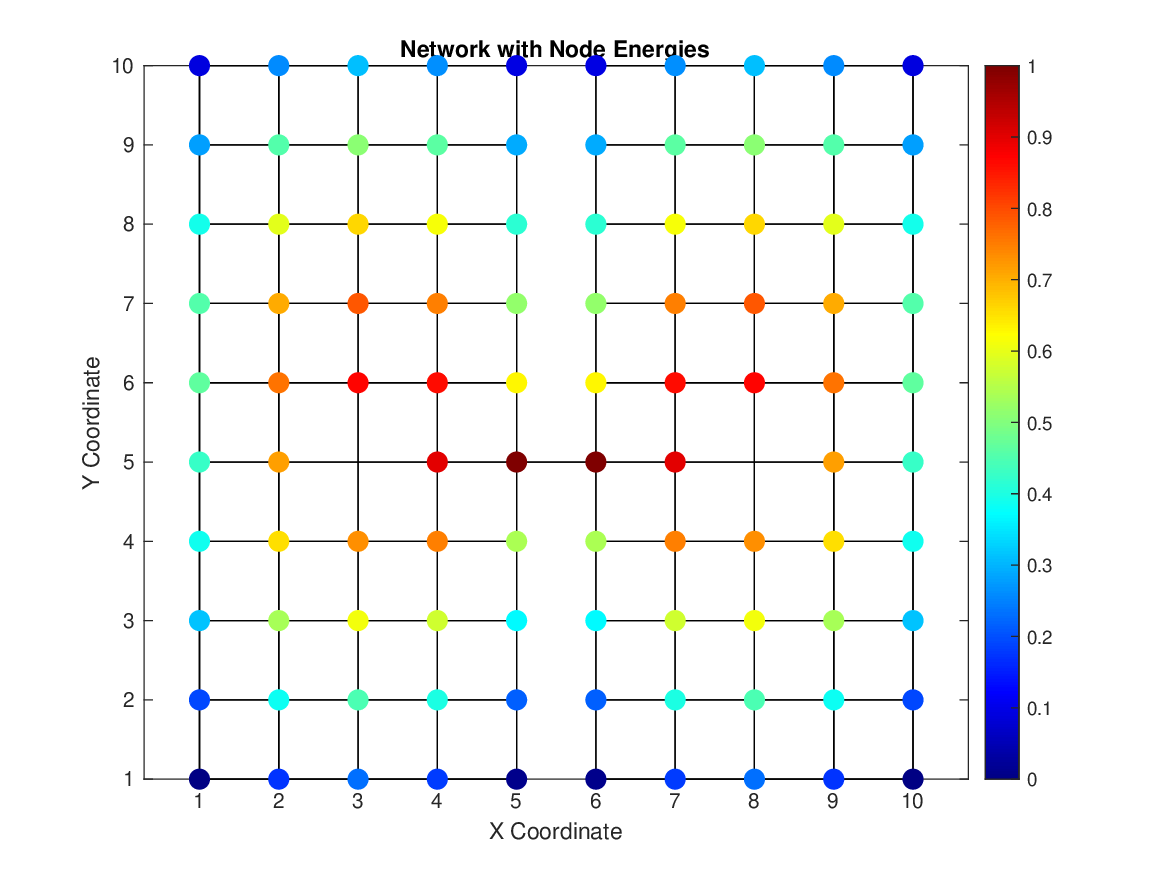} 
\includegraphics[width=.49\textwidth]{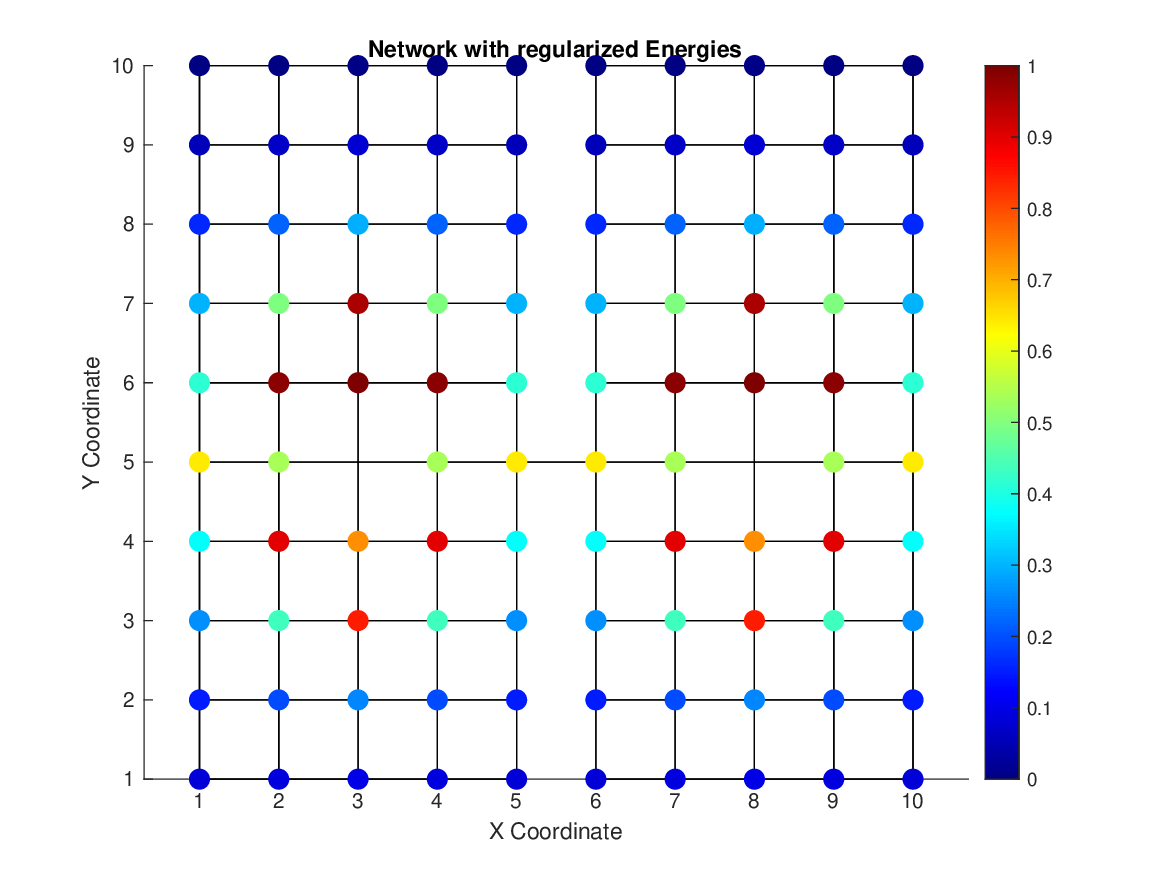}
    \caption{We consider an H-graph as described in \cref{s:CompExp3} and again consider the case where the first authority has chosen $Z_1$ and the second authority is choosing a vertex. 
    {\bf (top)} Here, $Z_1$ consists of a single node $(3,5)$ in the left subgraph. The left and right panels are as in \cref{fig:square}. 
    {\bf (bottom)} Here,  $Z_1$ consists of a two nodes, one taken from the left and one taken from the right subgraphs. Again, The left and right panels are as in \cref{fig:square}. See \cref{s:CompExp3} for details. }
    \label{fig:hgraph}
\end{figure}

\subsubsection{H-graph} \label{s:CompExp3}
We next consider an `H' graph, illustrated in \cref{fig:hgraph}. The graph is obtained by connecting two $5\times 10$ square grid graphs by a single edge (``bridge'') connecting the $(5,5)$ vertex of one with the $(1,5)$ vertex of the other (at (6,5) in the figure).  
In the top panels of \cref{fig:hgraph}, the first opinion authority has chosen $Z_1$ to be the $(3,5)$ node on the left subgraph. 
In the bottom panels, the first opinion authority has chosen $Z_1$ to be the $(3,5)$ node on both the left and right subgraphs. 
As before, in the left panels, we plot the energy that the second opinion authority would realize if they chosen each available node 
and, in the right panels, we plot the
the prediction $ \tilde \phi(i)$  via the relaxed problem \eqref{e:OptProb2Relax} with $\varepsilon = .15$.  
In both cases, the vertices in the top half of the graph are slightly preferred to the bottom half. 
In the left panels, we see that there is a trade-off between proximity to the set $Z_1$ and the high betweenness centrality of the vertices near the bridge.  Again, $\tilde \phi$ plotted in the right panels does a good job predicting this behavior.

\begin{figure}
    \centering
\includegraphics[width=.49\textwidth]{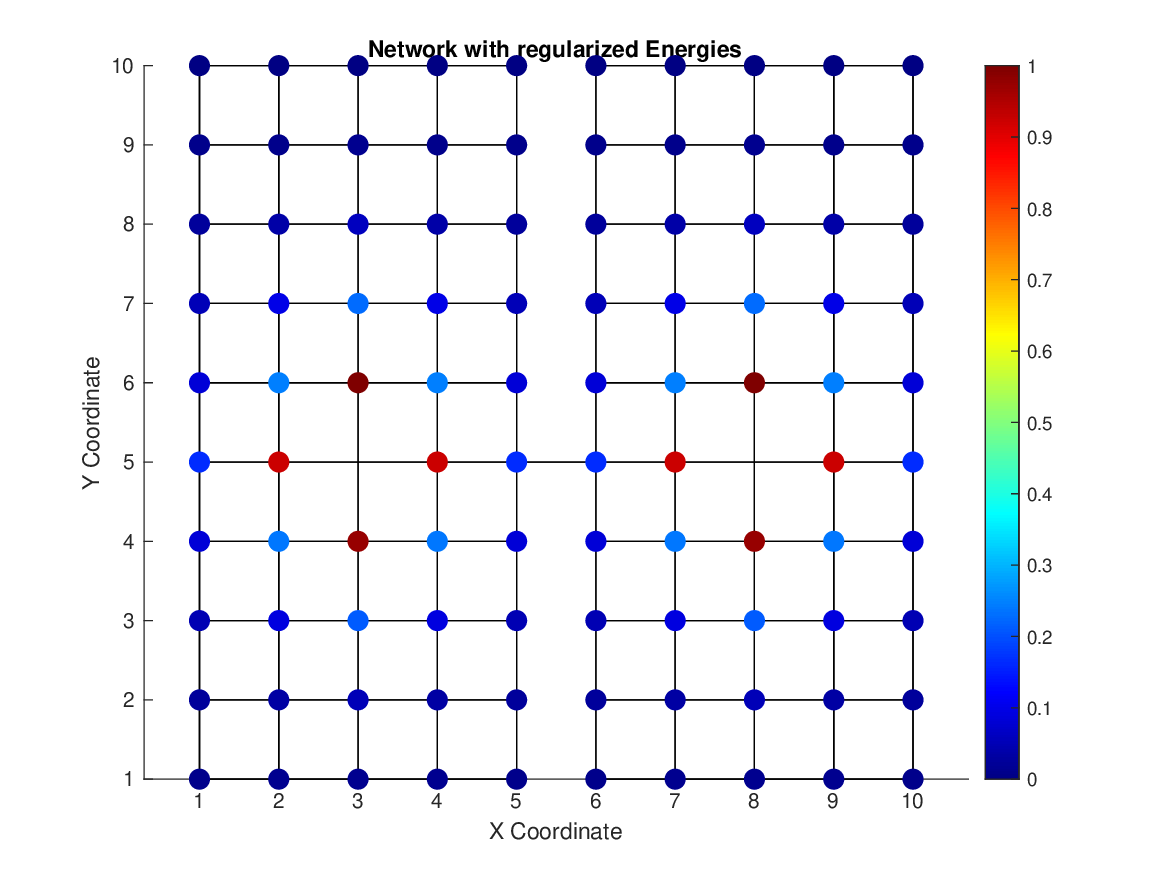}  
    \caption{As in \cref{fig:hgraph}, we consider an H-graph, but this time we plot $\tilde \phi$ for the value $\varepsilon = .015$, which is a factor of 10 smaller than the previous value. Compare to the bottom right figure of \cref{fig:hgraph}. We observe that the algorithm correctly wants to cordon off the first authority opinion. }
    \label{fig:smalleps}
\end{figure}

In a second experiment on the H-graph, we consider varying the parameter $\varepsilon >0$. 
We again consider the case where the first opinion authority has chosen $Z_1$ to be the $(3,5)$ nodes on both the left and right subgraphs. We now take $\varepsilon = .015$, a factor of 10 smaller than the previous value. 
We plot the function $\tilde \phi$ in \cref{fig:smalleps}, which should be compared to the bottom right panel of \cref{fig:hgraph}.  We observe that as $\varepsilon$ decreases, the function $\tilde \phi$ localizes near the set $Z_1$. In other words, as $\varepsilon \to 0$, $\phi$ will move to completely isolate the opinion class of $Z_1$.  Since $\phi$ is not really constrained to consider a budget of a small number of nodes, this is the optimal way to minimize $Z_1$'s influence.  However, for $\varepsilon$ chosen such that the global graph dynamics can be viewed, it can demonstrate further useful nodes for restricting spread of $Z_1$'s opinion.  An expanded version of our algorithm could actually be implemented to choose a response based upon a probability distribution set by the function $\phi$ with $\varepsilon$ chosen at a scale that allows for spreading out amongst the graph.

\subsubsection{Interactive game implementation} \label{s:CompExp4}
Finally, we implemented our algorithm, which can be found on github~\cite{gamegit}, in a graphical user interface at~\cite{gameweb}. The graphical user interface is designed for non-experts and allows the user to select a graph type and size. The graph types available are 
spatial random graphs, tree graphs, ladder graphs, square lattices, hexagonal lattices, triangle lattices\creflastconjunction cycles.  A version of the game played on an implementation of a random geometric graph with $50$ vertices is shown in \cref{fig:out2}. The user can also choose to have one or both players be played automatically. The automatic player can be set to random (easy), greedy (medium), and dynamic programming heuristics (hard) levels of play and are useful for developing intuition as well as for introducing the problem to popular audiences on a conceptual level.

\begin{figure}
    \centering
\includegraphics[width=.8\textwidth]{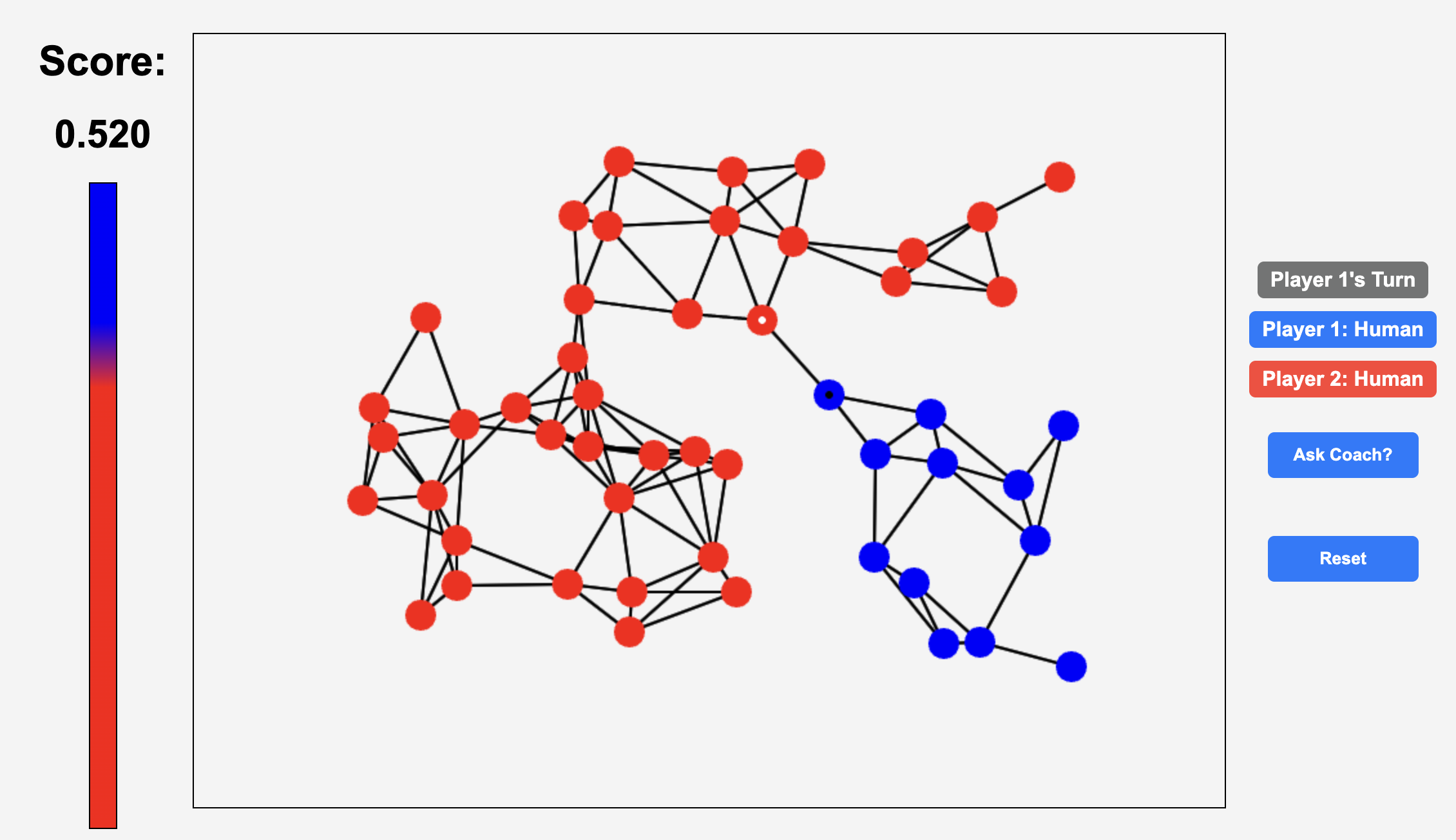} 
    \caption{An illustration of the interactive game implementation of the targeting influence problem  \cite{gameweb}. Here, we have an random geometric graph on 50 nodes. Two human players have each chosen a single node and Player 2 is currently winning; see \cref{s:CompExp4}. }
    \label{fig:out2}
\end{figure}

\section{Discussion} \label{sec:Disc}

In this paper, we introduced a model for targeting influence in a social network in settings where there are multiple extreme opinions; see \cref{e:OptProb}.  We first showed that each authority can group together the opposing opinions when considering how to maximize their own influence; see \eqref{e:OptProb2}. We established that the targeting influence problem is NP-hard to approximate; see \cref{t:NP-hard}. We also showed that the objective function in the influence targeting problem is monotone and submodular and hence admits a $1 - 1/e$ approximation algorithm; see \cref{t:submodular}. In \cref{sec:Relax}, we developed an approximation strategy for solutions of the  NP-hard targeting influence  problem. We show that the objective function for the relaxed optimization  problem \eqref{e:OptProb2Relax} is strictly concave and compute the gradient and Hessian; see \cref{t:Convexity}. We also proved that, in the case when the graph has a symmetry, the optimal solution to the relaxed problem also has the same symmetry; see \cref{p:GraphAuto}. Finally, in \cref{sec:Num}, we described a numerical implementation of our method and described the results of experiments to illustrate our results and test how well the relaxed optimization problem predicts the solution to the original problem.

There are several simplifications that we made in formulating the targeting influence  problem \cref{e:OptProb}, including the following. 

\begin{enumerate}
\item Each social connection is parameterized by a single number (the link weight), whereas real social connections are highly complex. Indeed, we might even disapprove of the choices made by our acquaintances; this could be modeled by incorporating signed weights in the directed graph or allowing for multiple kinds of connections. 

\item In the harmonic model for opinion dynamics used here, there is  no time dependence. This is equivalent to assuming that the timescale of social opinion spread is faster than the timescale on which players can select and influence particular nodes, which may or may not be valid, depending on the situation. The assumption of linear dynamics on this short time scale is also not obviously true, and is one of the main departures of our contribution from previous approaches, such as cascading models. A very interesting line of future work would be to explore further when the various dynamics of opinion spread are appropriate.

\item We have modeled zealots to hold only extreme opinions. While this might be a reasonable assumption in some cases (for example, in advertising, sponsored influencers tend to push a particular product), in political elections it might be better to allow multiple preferences of zealots (for example, an activist might prefer anyone but candidate X). 

\item We have assumed that an authority for a particular opinion can convert any chosen member to a zealot through sponsorship/incentives. But this might not be the case in all applications or the budget might have to account for different prices for conversion of different nodes. Our framework is extensible to such cases by incorporating costs for converting vertices,  i.e., 
in \eqref{e:OptProbRelaxC} or \eqref{e:OptProb2RelaxE}, we could assign weights $w_i \geq 0$, $i \in V \setminus Z$ and require $\sum_{i\in V} w_i \phi(i) = 1.$

\item In the numerical examples, we have explored the case when the different authorities play a game, taking turns adding to their respective zealot sets. However, in practice, authorities compete for influence all at once. 
\end{enumerate}

Beyond addressing the model shortcomings outlined above, there are of course other interesting possible future directions for this work. 
In the game scenario where the authorities take turns choosing a single zealot to convert, we think it is an interesting question of whether the first player has an advantage. 
In the case of an undirected graph, we conjecture that the first player, if they play optimally, will always beat or tie the second player. 
In the case of a directed graph, this is not true. Consider a directed cycle graph where all orientations are the same. If the second player chooses the node adjacent and ``downstream'' to the node chosen by the first player, they will completely influence all nodes but one and win the turn.

In recent years, there have been many developments investigating the consistency of problems posed on geometric graphs in the limit as the number of sampled nodes tends to infinity (see, e.g., \cite{calder2023rates,Osting_2017,YUAN_2021}). It would be interesting to consider this continuum limit here, where one might expect to obtain the Laplace problem on a fixed domain and authorities alternatively introduce boundary components where a Dirichlet condition is satisfied. A generalization that would allow Dirichlet data to be specified on lower dimensional sets is to consider the $p$-Laplacian for $p>2$. A very recent paper \cite{favre2024continuum} considers a continuum version of the continuous-time DeGroot model.

\section*{Acknowledgments}
The authors thank Wesley Hamilton for working with them on a very early version of this idea and Alex McAvoy for helpful conversations about game theory interpretations of these models. We also thank Wilson Stoddard, Connor McBride, Taylor Larkins, Ian Goodwin\creflastconjunction other members of the Boyd Lab for help developing the interactive software for the game \cite{gamegit}, which greatly helped develop our intuition. We gratefully acknowledge the Banff International Research Station and the Casa Mathematica Oaxaca for supporting our research efforts. 

\clearpage
\bibliographystyle{siamplain}
\bibliography{refs}
\end{document}